\documentclass[11pt,a4paper]{article}
\usepackage{etex}
\usepackage[utf8]{inputenc}
\usepackage[TS1,T1]{fontenc}
\usepackage{xspace}
\usepackage[a4paper,tmargin=3truecm,bmargin=3truecm,rmargin=2.2truecm,lmargin=2.2truecm]{geometry}
\usepackage{amssymb,amsmath,mathtools}
\usepackage[all]{xy}
\usepackage[inline]{enumitem}
\usepackage[english]{babel}
\usepackage{hyperref}
\usepackage{lmodern}
\usepackage{array,booktabs}

\usepackage{lineno}

\usepackage{amsthm,thmtools}
\newtheorem{theorem}{Theorem}[section]
\newtheorem{proposition}[theorem]{Proposition}

\newtheorem{corollary}[theorem]{Corollary}

\newtheorem*{mainlemma}{Main Lemma}

\theoremstyle{definition}

\theoremstyle{remark}

\hypersetup{
plainpages=false,
colorlinks=true,
linkcolor=black, 
anchorcolor=black, 
citecolor=red, 
urlcolor=black, 
menucolor=black, 
filecolor=black, 
bookmarksopen=true,
bookmarksnumbered=true}

\numberwithin{equation}{section}

\newcommand{\spmatrix}[1]{\left( \begin{smallmatrix}#1\end{smallmatrix}\right)} 

\DeclarePairedDelimiter\abs{\lvert}{\rvert}
\DeclarePairedDelimiter\norm{\lVert}{\rVert}

\newcommand\algzero{{\mathsf{0}}}
\newcommand\Id{{\text{\textup{Id}}}}
\newcommand\loc{{\text{\textup{loc}}}}
\newcommand\exter{{\textstyle\bigwedge}}

\newcommand\cdotaction{\mathord{\cdot}}

\newcommand{\dotomega}{\mathring{\omega}}
	
\newcommand{\dotA}{\mathring{A}}	
\newcommand{\dotF}{\mathring{F}}	
\newcommand{\dotW}{\mathring{W}}	
\newcommand{\dottau}{\mathring{\tau}}	
	
\newcommand{\dotvarphi}{\mathring{\varphi}}	

\newcommand\dd{\text{\textup{d}}}
\newcommand\upR{\text{\textup{R}}}

\newcommand{\homega}{\widehat{\omega}}
\newcommand{\hOmega}{\widehat{\Omega}}

\newcommand{\hphi}{\widehat{\phi}}
\newcommand{\hvarphi}{\widehat{\varphi}}
\newcommand{\hzeta}{\widehat{\zeta}}
\newcommand{\hGamma}{\widehat{\Gamma}}
\newcommand{\hA}{\widehat{A}}
\newcommand{\hF}{\widehat{F}}
\newcommand{\hR}{\widehat{R}}
\newcommand{\hD}{\widehat{D}}
\newcommand{\hcaD}{\widehat{\caD}}
\newcommand{\hd}{\widehat{\dd}} 
\newcommand{\hnabla}{{\widehat{\nabla}}}
\newcommand{\hg}{{\widehat{g}}}

\newcommand{\tu}{{\widetilde{u}}}
\newcommand{\tB}{{\widetilde{B}}}

\newcommand{\bGamma}{{\bar{\Gamma}}}
\newcommand{\bR}{{\bar{R}}}

\newcommand{\tF}{{\widetilde{F}}}
\newcommand{\ts}{{\widetilde{s}}}
\newcommand{\tvarphi}{{\widetilde{\varphi}}}
\newcommand{\teta}{{\widetilde{\eta}}}

\newcommand{\vc}{\vcentcolon =} 

\newcommand{\gt}{\rho_\caA}

\newcommand{\kg}{\mathfrak{g}}
\newcommand{\kh}{\mathfrak{h}}
\newcommand{\kj}{\mathfrak{j}}
\newcommand{\ku}{\mathfrak{u}}
\newcommand{\kk}{\mathfrak{k}}
\newcommand{\ksu}{\mathfrak{su}}
\newcommand{\kso}{\mathfrak{so}}
\newcommand{\kS}{\mathfrak{S}}

\newcommand\gR{{\mathbb R}}
\newcommand\gC{{\mathbb C}}

\newcommand\gZ{{\mathbb Z}}

\newcommand{\varE}{\mathcal{E}}
\newcommand{\varM}{\mathcal{M}}
\newcommand{\varP}{\mathcal{P}}
\newcommand{\varQ}{\mathcal{Q}}
\newcommand{\varU}{\mathcal{U}}

\newcommand{\caA}{\mathcal{A}}
\newcommand{\caF}{\mathcal{F}}
\newcommand{\caL}{\mathcal{L}}
\newcommand{\caG}{\mathcal{G}}
\newcommand{\caD}{\mathcal{D}}

\newcommand\lieA{{\mathbf{\mathsf{A}}}}
\newcommand\lieL{{\mathbf{\mathsf{L}}}}

\newcommand\sfX{{\mathsf X}}

\newcommand\sfR{{\mathsf R}}

\DeclareMathOperator{\Ad}{Ad}
\DeclareMathOperator{\End}{End}
\DeclareMathOperator{\diag}{diag}

\newcommand{\grast}{\bullet}

\newcommand{\auxfield}{auxiliary field\xspace}		
\newcommand{\auxfields}{auxiliary fields\xspace}		
\newcommand{\Auxfield}{Auxiliary field\xspace}		
\newcommand{\dressfield}{dressing field\xspace}		
\newcommand{\dressfields}{dressing fields\xspace}	
\newcommand{\Dressfield}{Dressing field\xspace}		
\newcommand{\resfield}{residual field\xspace}		
\newcommand{\resfields}{residual fields\xspace}		
\newcommand{\Resfield}{Residual field\xspace}		

\newcommand{\caseone}{case~1\xspace}
\newcommand{\caseonebis}{case~1b\xspace}
\newcommand{\Caseone}{Case~1\xspace}
\newcommand{\casetwo}{case~2\xspace}
\newcommand{\casetwobis}{case~2b\xspace}
\newcommand{\casethree}{case~3\xspace}

\hypersetup{
pdfauthor={C. Fournel, J. François, S. Lazzarini, T. Masson},
pdftitle={Gauge invariant composite fields out of connections, with examples},
pdfsubject={},
pdfcreator={pdflatex},
pdfproducer={pdflatex},
pdfkeywords={}
}

\begin{document}

\title{Gauge invariant composite fields out of connections,\\ with examples}
\author{C. Fournel, J. François, S. Lazzarini, T. Masson}
\date{}
\maketitle
\begin{center}
Aix Marseille Université, CNRS, CPT, UMR 7332, Case 907, 13288 Marseille, France\\
Université de Toulon, CNRS, CPT, UMR 7332, 83957 La Garde, France
\end{center}

\begin{abstract}
In this paper we put forward a systematic and unifying approach to construct gauge invariant composite fields out of connections. It relies on the existence in the theory of a group valued field with a prescribed gauge transformation. As an illustration, we detail some examples. Two of them are based on known results: the first one provides a reinterpretation of the symmetry breaking mechanism  of the electroweak part of the Standard Model of particle physics; the second one is an application to Einstein's theory of gravity described as a gauge theory in terms of Cartan connections. The last example depicts a new situation: starting with a gauge field theory on Atiyah Lie algebroids, the gauge invariant composite fields describe massive vector fields. Some mathematical and physical discussions illustrate and highlight the relevance  and the generality of this approach.
\end{abstract}

\section{Introduction and motivations}
\label{sec-introduction}

In gauge field theories, gauge symmetries are redundant inner degrees of freedom which have to be managed, for instance, to define gauge invariant observables, or to proceed to the quantization of these theories. A distinguished method to deal with this problem consists in performing a reduction of symmetries. There are three familiar procedures to do so, which answer different issues in gauge field theories.

Firstly, gauge fixing is a technique which allows to simplify field equations by avoiding physically equivalent field configurations. This point is of particular importance for the quantization of field theories in order to mod out the volume of the gauge group in the functional integral. To perform this procedure, one selects a single representative in the gauge orbit of the fields by adding a constraint equation either in the functional measure or directly in the Lagrangian. This additional term prevents the action from being gauge invariant anymore, so that the symmetry is \textsl{de facto} reduced. 

Secondly, the spontaneous symmetry breaking mechanism has been devised in \cite{EnglBrou64a, Higg64a, GuraHageKibb64a}. This procedure requires an external scalar field coupled to the gauge fields of the theory. At certain values of a parameter in the theory, the scalar field is spontaneously polarized in a direction which minimizes its potential. By doing so in the electroweak part of the Standard Model, the system undergoes a ``phase transition'' from a massless theory to a theory with massive vector bosons, $W_\mu^\pm$ and $Z_\mu$, mediating the weak interaction. The symmetry group is broken into a residual subgroup which, as a requirement of this mechanism, leaves invariant  the vacuum configuration.

Thirdly, the procedure of reduction of principal fiber bundles results from a mathematical theorem in the theory of  fiber bundles \cite{KobaNomi96a}: a $G$-principal fiber bundle $\varP$ is reduced to a $H$-principal fiber bundle $\varQ$, where $H \subset G$ is a Lie subgroup, if and only if there exists a global section of the $G/H$-bundle $\varP/H$. For instance, the metric $g$, in general relativity, induces a reduction of the structure group $GL^+_4(\gR)$ to its subgroup $SO(1,3)$.

The procedure described in this paper implements a reduction of symmetry, but it does not belong to any of the latter sorts. It can be applied to any gauge theory equipped with a group valued field with a specific action of the gauge group, which we refer to as a ``\dressfield''.  It is a mathematical procedure in the sense that it does not depend on either a parameter in the theory or a convenient constraint equation. It consists in constructing gauge invariant composite fields out of connections and \dressfields: this is merely an appropriate change of variables in the functional spaces of fields. The action of  the gauge group is ``neutralized'' in the sense that it becomes trivial on these composite fields. Some examples, usually described within distinct frameworks, are treated here in the general structure explicitly depicted in the ``Main Lemma'' of section~\ref{sec-general-scheme}. In practice, the \dressfield comes from what will be called an ``\auxfield'' in the gauge theory, whose Lagrangian, once written in terms of the new fields, depends only on the gauge invariant composite fields, and on what will be called ``\resfields'', and not anymore on the \dressfield.

\smallskip
In this paper, a gauge theory is a field theory supporting the action of the gauge group $\caG$ of a principal fiber bundle $\varP$, which is the space of vertical automorphisms of $\varP$. Let us stress the difference between \emph{active} and \emph{passive} gauge transformations. The former correspond to transformations of the fields by elements of $\caG$, usually implemented through the geometric action of $\caG$ on $\varP$. The latter correspond to changes of  local trivializations of $\varP$. In most gauge field theories, the mathematical expressions of both active and passive transformations look the same. However, we present in section~\ref{sec-application-to-a-yang-mill-higgs-model-based-on-atiyah-lie-algebroids} an example where this situation is no longer true.

Gauge symmetries rely on inner degrees of freedom, and will be distinguished from geometrical symmetries induced by changes of coordinate systems or action of diffeomorphisms on a manifold. This distinction makes sense in relation to the notion of observables in physics. Namely, using the present meaning of gauge transformations, only gauge invariant quantities can be observed, but there are lots of coordinate dependent quantities which can be observed: position, momentum, electromagnetic currents, electric and magnetic fields, \textsl{etc}. This clear separation between inner and geometric degrees of freedom is also apparent from a mathematical point of view.  Geometrical structures are related to the theory of natural fibre bundles (see \textsl{e.g.} \cite[Section~14]{KolaMichSlov93a}), where they are naturally associated to a smooth manifold and its diffeomorphism group. On the contrary, gauge field theories require extra structures in the form of a structure group,  a principal fiber bundle, and some connections… Accordingly, we consider a linear connection on a manifold as a natural geometrical object, defined on the natural fiber bundles of tensor fields.

The scheme put forward in this paper brings out a procedure of \emph{geometrization}, in that it transforms fields defined in a gauge theory, on a principal fiber bundle or on an associated bundle, into fields defined in terms of the natural geometry of the base manifold. In other words, it gets rid of some of the extra structures required by a gauge field theory to the benefit of some geometrical objects, which turn out to enter in the construction of physical observables. Nevertheless, it is not always possible to perform a full geometrization of a gauge theory.

\smallskip
In section~\ref{sec-general-scheme}, we present the scheme which leads to the construction of gauge invariant composite fields out of a connection. The technical component is summarized in a lemma, that we refer to as the ``Main Lemma'' of this paper. Very simple illustrative examples are proposed. A general mathematical discussion on the relevance of the actions of the gauge group, as well as on the geometry of the \dressfield, which is an essential ingredient in the scheme. 

In section~\ref{sec-applications-to-gauge-theories-in-particles-physics},  we illustrate more substantially our scheme through the electroweak part of the Standard Model of particle physics. The \dressfield is extracted from the scalar field used to perform the symmetry breaking mechanism in the usual point of view. The composite fields $W_\mu^\pm$, $Z_\mu$ and $A_\mu$ are exactly the fields of the bosons of the Standard Model after symmetry breaking. An extension of this method to the case $SU(N)$ is also explored for $N>2$: a direct application of the Main Lemma is no longer possible, still, a reduction from $SU(N)$ to $SU(N-1)$ is described.

In section~\ref{sec-application-to-general-relativity-as-a-gauge-theory}, we consider general relativity (GR). Its original formulation by Einstein is given in terms of geometric structures only. It is also possible to consider GR as a gauge theory, but the Lagrangian is not of Yang-Mills type (see \textsl{e.g.} \cite{Trau79a}), and the soldering form plays an important role to recover the geometric theory out of the gauge theory. Here, we use a Cartan connection on a $SO(1,m-1)$-principal fiber bundle. The \dressfield is extracted from the ``translational'' part of the connexion, and plays the role of a vielbein. The gauge invariant composite fields behave exactly as Christoffel symbols.

In Section~\ref{sec-application-to-a-yang-mill-higgs-model-based-on-atiyah-lie-algebroids}, we develop a completely new example which is so far a toy model. It is constructed using the general mathematical framework of gauge theories on transitive Lie algebroids which has been developed recently in  \cite{FourLazzMass13a}. In this specific example, we consider generalized connections on the Atiyah Lie algebroid of a $G$-principal fiber bundle $\varP$ as a natural background for Yang-Mills-Higgs theories, where the action of the gauge group is not the usual geometric action. To construct this model, we consider only a subspace of generalized connections for which a \dressfield is automatically provided. From the application of the Main Lemma, it follows that the gauge invariant composite fields are massive vector fields.

In section~\ref{sec-comments}, we use our unifying scheme to make comparisons between the fields involved in the three main examples, see Table~\ref{table-before-MainLemma} and Table~\ref{table-after-MainLemma}, and we discuss similarities and differences with some other usual constructions. This gives a better positioning of this pretty appealing scheme in the landscape.


\section{The scheme and first illustrative examples}
\label{sec-general-scheme}

Let us consider a gauge field theory, with structure Lie group $H$, with Lie algebra $\kh$. Let $\omega$ be a connection, that is, a $1$-form on a smooth manifold with values in $\kh$ which varies under a gauge transformation $\gamma$ (function with values in $H$), as $\omega^\gamma = \gamma^{-1} \omega \gamma + \gamma^{-1} \dd \gamma$.

\begin{mainlemma}
Suppose that there exists a field $u$ with values in a Lie group $G$ containing $H$, such that under a gauge transformation $\gamma$ one has $u^\gamma = \gamma^{-1} u$. Then the composite field
\begin{equation}
\label{eq-compositefield}
\homega \vc u^{-1} \omega u + u^{-1} \dd u,
\end{equation}
if it is mathematically defined, is $H$-gauge invariant.

If $\phi$ is a vector field in a representation $\ell$ of $H$, which varies under a gauge transformation $\gamma$ as $\phi^\gamma = \ell(\gamma^{-1}) \phi$, then the composite field $\hphi \vc \ell(u^{-1}) \phi$, if it is mathematically defined, is $H$-gauge invariant.

Let $\Omega=\dd \omega + \frac{1}{2} [\omega,\omega]$ be the curvature of $\omega$. Then the composite field $\hOmega \vc u^{-1}\Omega u$, if it is mathematically defined, is gauge invariant, and one has $\hOmega = \dd \homega + \frac{1}{2} [\homega,\homega]$.

Let $\caD \phi \vc \dd \phi + \ell_*(\omega)\phi$ be the covariant derivative of $\phi$ as above associated to $\omega$. Then the composite field $\widehat{\caD \phi} \vc \ell(u^{-1}) \caD \phi$, if it is mathematically defined, is gauge invariant, and, defining $\widehat{\caD \phi} \vc \hcaD \hphi$, one has $\hcaD = \ell(u^{-1}) \caD \ell(u) = \dd + \ell_*(\homega)$.
\end{mainlemma}

\begin{proof}
One has $\homega^\gamma \vc (u^\gamma)^{-1} \omega^\gamma u^\gamma + (u^\gamma)^{-1} \dd u^\gamma = (u^{-1} \gamma) (\gamma^{-1} \omega \gamma + \gamma^{-1} \dd \gamma) (\gamma^{-1} u) + u^{-1} \gamma \dd (\gamma^{-1} u) = u^{-1} \omega u + u^{-1} \dd u$. A similar computation holds for $\hphi$, $\hOmega$, and $\widehat{\caD \phi}$, and the explicit formulas of $\hOmega$ and $\hcaD$ are obtained straightforwardly.
\end{proof}

In the paper we will refer to the field $u$ of the Main Lemma as the \emph{\dressfield}. The mention ``if it is mathematically defined'' means that the composite field should be defined without ambiguity in an identified space of fields. Examples will be clear enough to illustrate this point. The lemma is not precise concerning the space of connections $\omega$. Indeed, it works for descriptions of connections either as $1$-forms on a $H$-principal  fiber bundle $\varP$ over a manifold $\varM$, or as local $1$-forms over an open subset $\varU \subset \varM$ which trivializes $\varP$. In the same way, the field $\phi$ could be an equivariant field on $\varP$ or a local trivialization of a field on $\varU$. This fuzziness is convenient to apply this result to many examples, possibly with minor adjustments, and to look at quite different situations in an unifying point of view. The differential operator $\dd$ can be more general than the de~Rham differential, provided the corresponding terms are mathematically well-defined. The third example of this paper (see section~\ref{sec-application-to-a-yang-mill-higgs-model-based-on-atiyah-lie-algebroids}) uses such a differential calculus extending the de~Rham differential calculus. 

\medskip
As a first illustrative application, consider the following prototype Stueckelberg Lagrangian~\cite{RuegRuiz04a}
\begin{equation*}
\caL[A_\mu, B]= -\tfrac{1}{4}  F_{\mu\nu}F^{\mu\nu} + \tfrac{m^2}{2} \left(A_\mu -\tfrac{1}{m}\partial_\mu B\right)\left( A^\mu -\tfrac{1}{m}\partial^\mu B \right),
\end{equation*}
where $F_{\mu\nu } = \partial_\mu A_\nu - \partial_\nu A_\mu$ is the field strength associated to the $U(1)$-connection $A_\mu$,  $B$ is the Stueckelberg field with values in the Lie algebra $\ku(1)$ and $m$ is a constant parameter. This Lagrangian is invariant under the infinitesimal gauge transformations $\delta A_\mu=-\partial_\mu \Lambda$ and $\delta B=-m\Lambda$ for any $\ku(1)$-valued function $\Lambda$. Consider the $U(1)$-valued \dressfield $u=e^{ i B/m}$ which transforms as $u^\gamma = e^{i (B-m\Lambda)/m}=\gamma^{-1}u$ where $\gamma=e^{i\Lambda}$ is an element of the $U(1)$-gauge group. Applying the Main Lemma, the gauge invariant composite fields are  $\hA_\mu=A_\mu + i u^{-1}\partial_\mu u=A_\mu - \tfrac{1}{m}\partial_\mu B$ and $\hF_{\mu\nu} = F_{\mu\nu}$. Then, the Stueckelberg Lagrangian reduces to
 \begin{equation*}
\mathcal{L}[\hA_\mu]= -\tfrac{1}{4}  \hF_{\mu\nu}\hF^{\mu\nu} + \tfrac{m^2}{2} \hA_\mu\hA^\mu.
\end{equation*}
This is a Proca-like Lagrangian describing a gauge invariant massive vector field $\hA_\mu$. Schematically, we shifted a $U(1)$-gauge theory with fields $(A_\mu, B)$ to a theory built on the gauge invariant composite fields $\hA_\mu$ in which the $U(1)$ symmetry has been factorized out.

Let us recall that the so-called Stueckelberg trick consists precisely in the opposite shift, that is to implement a $U(1)$-gauge symmetry on the Proca Lagrangian at the expense of introducing a new field $B$. The degree of freedom of the field $B$ is introduced in the Lagrangian in order to exactly compensate for the enlargement to $U(1)$-gauge symmetry. On the contrary, in our scheme, the degrees of freedom of the \dressfield are absorbed in order to reduce the size of the gauge group.


The scheme that will be repeatedly used in the next sections consists in identifying the \dressfields as a part of some \auxfields given from the very beginning for free in the gauge theory at hand. In the Stueckelberg Lagrangian, the \dressfield and the \auxfield coincide. To give a less trivial example of this scheme, consider the abelian Higgs model
\begin{equation*}
\caL[A_\mu,\varphi] = \left[ (\partial_\mu -i A_\mu) \varphi\right]^\dagger \left[ (\partial_\mu -i A_\mu) \varphi \right] - V(\varphi) - \tfrac{1}{4}  F_{\mu\nu}F^{\mu\nu}
\end{equation*}
where $\varphi$ is a $\gC$-valued scalar field, and $V(\varphi)= \mu^2 \varphi^\dagger \varphi + \lambda (\varphi^\dagger \varphi)^2$. This Lagrangian is invariant under the (finite) gauge transformations $A_\mu^\gamma = A_\mu + i \gamma^{-1} \partial_\mu \gamma = A_\mu - \partial_\mu \alpha$ and $\varphi^\gamma = \gamma^{-1} \varphi = e^{-i \alpha} \varphi$ for $\gamma = e^{i \alpha}$. The \dressfield $u$ is identified from the \auxfield $\varphi \neq 0$ by the polar decomposition $\varphi = \rho u$, where $\rho= |\varphi|$. According to $\varphi^\gamma = \rho^\gamma u^\gamma = \gamma^{-1} \rho u$, the $U(1)$-valued field $u$ transforms as $u^\gamma =  \gamma ^{-1} u$ and $\rho$ is invariant. Applying the Main Lemma, the composite fields $\hA_\mu = A_\mu + i u^{-1}\partial_\mu u$ and $\widehat \varphi = u^{-1}\varphi = \rho$ are gauge-invariant fields, $\hF_{\mu\nu}= F_{\mu\nu}$ and $V(\varphi)=V(\rho)$. The Lagrangian can be rewritten as:
\begin{equation*}
\caL[\hA_\mu,\rho] = \left[ (\partial_\mu - i \hA_\mu) \rho \right]^\dagger \left[ (\partial^\mu - i \hA^\mu) \rho \right] - V(\rho) - \tfrac{1}{4}  \hF_{\mu\nu}\hF^{\mu\nu}
\end{equation*}
This theory describes a (massless) vector boson, coupled with a $\gR^+$-valued scalar residual field~$\rho$, embedded in a potential term. Here, the representation of the gauge group is trivial on every fields of the theory. The gauge-invariant composite fields $\hA_\mu$ are constructed using  fields already present in the Lagrangian, so that the change of variables $(A_\mu,\phi) \mapsto (\hA_\mu,\rho)$ is merely a convenient redistribution of the degrees of freedom of the original theory.

\medskip
The idea to construct gauge invariant fields by ``dressing'' the fields of the theory as in the Main Lemma, takes its root in  \cite{Dira55a} for QED (see also \cite{Dira58a}), and has been used in \cite{LaveMcMu97a} for QCD, where the terminology ``dressing  field'' is explicitly introduced. Relation with our scheme is postponed to \ref{sec-about-the-dressing-field}. But let us mention right now that the dressing fields exhibited in these papers are non-local with respect to the gauge field variables. In our forthcoming three main examples, the dressing fields are local, so that the composite fields are also local.

\medskip
Let us now make some comments about the Main Lemma. A gauge field theory requires that the fields belong to $\caG$-spaces, that is, spaces carrying specified actions of the gauge group of the theory. A space of fields $\caF$ can carry two different actions $\rho$ and $\rho'$ of the gauge group, but a field cannot belong to the two different $\caG$-spaces $(\caF, \rho)$ and $(\caF, \rho')$ at the same time. 

Let us recall the usual geometrical setting for the action of the gauge group in gauge field theories, in terms of principal fiber bundles and connections \cite{Ster94a}. Let $\varP$ be a $H$-principal  fiber bundle. A gauge transformation identifies with an equivariant map $\gamma : \varP \rightarrow H$ such that $\gamma(ph) = h^{-1} \gamma(p) h$ for any $p\in \varP$ and $h \in H$, and it defines a vertical automorphism of principal fiber bundle $\Psi : \varP \rightarrow \varP$, defined by $\Psi(p) \vc p \gamma(p)$. We denote by $\caG$ the gauge group of $\varP$. This group acts (on the right) on various spaces by a \emph{geometric} action induced as the pull-back by vertical automorphisms. Let us produce some well-known examples. Identifying a connection as a $1$-form $\omega \in \Omega^1(\varP) \otimes \kg$, this geometrical action gives rise to the usual formula $\omega \mapsto \omega^\gamma \vc \Psi^\ast \omega = \gamma^{-1} \omega \gamma + \gamma^{-1} \dd \gamma$. We denote by $(\caA, \gt)$ the $\caG$-space of connections carrying this action of $\caG$. As a second example, let $E \vc \varP \times_\ell F$ be an associated fiber bundle, where $\ell$ is a left action of $H$ on the fiber $F$. It is known that a section $s \in \Gamma(E)$ identifies with an equivariant map $\ts : \varP \rightarrow F$, such that $\ts(ph) = \ell_{h^{-1}} \ts(p)$. Then $\gamma \in \caG$ acts on $\Gamma(E)$ by the pull-back $\ts \mapsto \ts^\gamma \vc \Psi^\ast \ts$, so that $\ts^\gamma(p) = \ell_{\gamma(p)^{-1}} \ts(p)$. The geometric action of $\caG$ on any associated fiber bundle $E$ as above, is structurally written in terms of the action $\ell$. 

Until now, the space $\caG$ has been considered as the space of symmetries of the theory, \textsl{i.e.} $\caG$ is a group acting on fields. In gauge field theories, it is also possible to consider $\caG$ as a space of fields, which then requires the specification of an action of the gauge group $\caG$. The distinction between $\gamma \in \caG$, group element, and $u \in \caG$, field, stems from physical arguments. Two actions can be naturally defined on the space $\caG$. One has $\caG = \Gamma(\varP\times_{\alpha} H)$ where $\alpha_h(h') = h h' h^{-1}$ for any $h,h' \in H$. From this we deduce that the \emph{geometric} action of $\gamma \in \caG$ on any $u \in \Gamma(\varP\times_{\alpha} H)$ is given by $u^\gamma = \gamma^{-1} u \gamma$. We denote by $(\caG, \alpha)$ the $\caG$-space of fields $\caG$ carrying this action $\alpha$ of $\caG$. Another action of $\gamma \in \caG$ (gauge group) on $u \in \caG$ (space of fields) is given by $u^\gamma \vc \gamma^{-1} u$, which uses the product in the group. We denote by $(\caG, \sfR)$ the $\caG$-space of fields $\caG$ carrying this \emph{algebraic} action. The product in $\caG$ (space of fields) is compatible with $\alpha$ but not with $\sfR$.


It is common in gauge field theories to combine several fields living in different $\caG$-spaces into \emph{composite} fields. The action of $\caG$ on these composite fields is computed by combining the gauge transformed individual fields. For instance, the combination $\caD \phi \vc \dd \phi + \ell_*(\omega)\phi$ is such a composite field, whose gauge transformation (action of $\caG$) is given by $(\caD \phi)^\gamma \vc \dd \phi^\gamma + \ell_*(\omega^\gamma)\phi^\gamma$. In physics, the gauge principle requires that $(\caD \phi)^\gamma = \ell_{\gamma^{-1}} \caD \phi$, which is automatic in the present geometric setting. In general, any composite fields built on fields carrying geometrical actions of $\caG$ carries in turn a geometrical action. As a second example, consider $u \in (\caG, \alpha)$ (space of fields) and $\omega \in (\caA, \gt)$ and the composite field $\omega^u \vc u^{-1} \omega u + u^{-1} \dd u$ (which is not a gauge transformation: $u$ is a field!). Then, for any $\gamma \in \caG$ (gauge group), one has $(\omega^u)^\gamma \vc (\gamma^{-1} \omega \gamma + \gamma^{-1} \dd \gamma)^{\gamma^{-1} u \gamma} = \gamma^{-1} \omega^u \gamma + \gamma^{-1} \dd \gamma$, which shows that $\omega^u \in (\caA, \gt)$. This defines a map of $\caG$-spaces
\begin{equation}
\label{eq-action-G-spaces-one}
(\caG, \alpha) \times (\caA, \gt) \to (\caA, \gt).
\end{equation}
The usual relation $(\omega^{\gamma_1})^{\gamma_2} = \omega^{\gamma_1 \gamma_2}$ expresses the right action of $\caG$ (gauge group) on $\caA$ for $\gamma_1, \gamma_2 \in \caG$. Here, $\omega^{\gamma_1}$ is not a composite field since $\gamma_1$ is not a field.

Suppose now that we want to apply the Main Lemma in the situation when $u \in \caG = \Gamma(\varP\times_{\alpha} H)$. Then the action of $\caG$ required for a \dressfield, $u \mapsto \gamma^{-1} u$, is not the geometric action $\alpha$ carried by $\caG$: it is the algebraic action $\sfR$. This implies that the three hypotheses 
\begin{enumerate*}[label=(\roman*),itemjoin*={{ and }}]
\item\label{list-hyp1} geometric action of the gauge group,
\item\label{list-hyp2} $u$ a field in the space $\caG = \Gamma(\varP\times_{\alpha} H)$ (without specifying any action),
\item\label{list-hyp3} $u$ a \dressfield (\textsl{i.e.} $u$ transforms as $u \mapsto \gamma^{-1}u$),
\end{enumerate*}
can not hold at the same time. Indeed, \ref{list-hyp1}  $+$ \ref{list-hyp2} leads to $u \in (\caG, \alpha)$, while \ref{list-hyp2}  $+$ \ref{list-hyp3} leads to $u \in (\caG, \sfR)$. In this particular situation, the Main Lemma shows that \ref{list-hyp2}  $+$ \ref{list-hyp3} defines a map of $\caG$-spaces
\begin{equation}
\label{eq-action-G-spaces-two}
(\caG, \sfR) \times (\caA, \gt) \to (\caA, \Id).
\end{equation}
where $\Id$ is the trivial action of $\caG$ on $\caA$. In general, the Main Lemma would produce a composite fields $\homega$ outside the initial space of fields $\caA$ (see \ref{sec-an-alternative-dressing-field-particles} and \ref{sec-application-of-the-main-lemma-gravity}). But here, the assumption $u \in (\caG, \sfR)$ implies that $\homega \in \caA$, with the representation $\Id$. Had we forgotten the $\caG$-actions on each of the spaces, \eqref{eq-action-G-spaces-one} and \eqref{eq-action-G-spaces-two} would have been reduced to the same map of spaces (and no more of $\caG$-spaces): $\caG \times \caA \to \caA$. This would have brought some confusion on the true nature of the involved objects.

The hypothesis \ref{list-hyp1} looks imperative in the ordinary geometric setting of gauge field theories, but the third example in section~\ref{sec-application-to-a-yang-mill-higgs-model-based-on-atiyah-lie-algebroids} shows that this hypothesis can be bypassed.

\smallskip
Several arguments can be used to emphasize that relation \eqref{eq-compositefield} is not a gauge transformation, the first one being that \eqref{eq-compositefield} defines a composite field. The second one being that the \dressfield $u$ may not even be an element in the space $\caG$  (see examples in sections~\ref{sec-applications-to-gauge-theories-in-particles-physics} and \ref{sec-application-to-general-relativity-as-a-gauge-theory}), and the gauge invariant composite fields $\homega$ and $\hOmega$ may not be $\kh$-valued anymore, for instance when $G \neq H$ (see section~\ref{sec-application-to-general-relativity-as-a-gauge-theory}). In the same way, if $\phi$ is a section of an associated vector bundle to $\varP$, then $\hphi$ and $\hcaD\hphi$ needs not be a section of this vector bundle anymore.

Retaining hypotheses \ref{list-hyp1} (natural hypothesis in the geometric setting of gauge field theories) and \ref{list-hyp3}, we can conclude two facts. Firstly, \eqref{eq-compositefield} is definitively not a gauge transformation, because \ref{list-hyp2} is false, \textsl{i.e.} $u$ is not an element of the space $\caG$. Secondly, when $u$ takes its values in $H$, the action $u \mapsto \gamma^{-1} u$  being geometric by \ref{list-hyp1}, $u \in \Gamma(\varP \times_L H) \simeq \Gamma(\varP)$ with $L_h (h') = hh'$. It is known that such a global section exists if, and only if, $\varP \simeq \varM \times H$, so that the global existence of a \dressfield satisfying \ref{list-hyp1} and \ref{list-hyp3} implies strong requirements on the topology of $\varP$. We now elaborate on this specific point. 

Let $\varP$ be a $H$-principal fiber bundle on $\varM$, and let $K$ be a Lie subgroup of $H$ with Lie algebra $\kk \subset \kh$.
\begin{proposition}
\label{prop-S-function-topology}
There exists a map $S : \varP \rightarrow K$, such that $S(pk) = k^{-1} S(p)$ for any $p \in \varP$ and $k \in K$ if, and only if, there is an isomorphism of $K$-spaces $\varP \simeq \varP/K \times K$ where the (right) action of $K$ on $\varP/K \times K$ concerns only the $K$ factor.
\end{proposition}

\begin{proof}
If there is an isomorphism $\varP \xrightarrow{\simeq} \varP/K \times K$  of $K$-spaces, with $p \mapsto ([p]_K, k)$, then the map $S(p) = k^{-1}$ satisfies the requirements.

Suppose now that the map $S$ exists. Then $S$ is onto: for any $k \in K$, consider any $p_0 \in \varP$, then $S(p_0 S(p_0) k^{-1}) = k$. We can then define the non empty space $\varQ \vc S^{-1}(\{e\}) \subset \varP$, where $e$ is the unit in $K$. Then the map $\varP \rightarrow \varQ \times K$, defined by $p \mapsto (p S(p), S(p)^{-1})$, is a $K$-equivariant isomorphism, whose inverse is $\varQ \times K \ni (q,k) \mapsto qk \in \varP$.

The map $\varP \rightarrow \varQ$, defined by $p \mapsto p S(p)$, factorizes through the quotient $\varP \rightarrow \varP/K$, $p \mapsto [p]_K$, so that one has a map which associates to $[p]_K \in \varP/K$ the element $p S(p) \in \varQ$. This is an isomorphism, with inverse $q \mapsto [q]_K$ for any $q \in \varQ \subset \varP$.
\end{proof}

Notice that the map $S$ permits to realize the quotient $\varP/K$ as the subspace $\varQ$ of $\varP$. When $K = H$,  $S$ defines a global section of $\varP$, and one has $\varP/H = \varM$, so that $\varP \simeq \varM \times H$ as expected. In general, the proposition tells us that the existence of $S$ implies that $\varP$ is ``trivial in the $K$ direction''. 

Let $\omega$ be a connection $1$-form on $\varP$, and let $\homega \vc S^{-1} \omega S + S^{-1} \dd S$ be the composite field of the Main Lemma for the \dressfield $S$.  Define $f_S : \varP \rightarrow \varP$ by $f_S(p) = pS(p)$ for any $p\in \varP$. Then it is straightforward to show that $\homega = f_S^\ast \omega$ and $\homega$ is $K$-invariant and $K$-horizontal, so that $\homega$ defines a natural $1$-form on $\varP/K$. In the case $K=H$, we have a ``full geometrization'' of $\omega$ into the global $\kh$-valued $1$-form $\homega$ on $\varM$.

\begin{corollary}
\label{cor-case-J-times-K}
Suppose that a map $S$ as in Prop.~\ref{prop-S-function-topology} exists and that $H = J \times K$, with Lie algebra $\kh = \kj \oplus \kk$. Then $\varQ \vc S^{-1}(\{e\}) \subset \varP$  is a $J$-principal fiber bundle, and $\homega \vc S^{-1} \omega S + S^{-1} \dd S$ defines a $1$-form $\homega_\varQ = \omega_\varQ^\kj \oplus \homega_\varQ^\kk$ on $\varQ$, where $\omega_\varQ^\kj$ is a $J$-connection and $\homega_\varQ^\kk$ is a $\kk$-valued $K$-gauge invariant $1$-form.
\end{corollary}

\begin{proof}
Using the proof of Prop.~\ref{prop-S-function-topology}, one has $\varQ \simeq \varP/K = \varP \times_L H/K = \varP \times_{L'} J$ where $L$ is the induced left action of $H$ on $H/K = J$, and $L'_{(j,k)}(j') = j j'$, so that $\varQ$ is a $J$-principal fiber bundle. Using $\kh = \kj \oplus \kk$, one has $\omega = \omega^\kj \oplus \omega^\kk$, and $\homega = \omega^\kj \oplus \homega^\kk$, because $S$ is $K$-valued. In this decomposition, $\omega^\kj$ is connection-like for gauge transformations along $J$, and $\homega^\kk$ is gauge invariant along $K$. Let $\iota : \varQ \rightarrow \varP$ be the inclusion. Then $\omega_\varQ^\kj = \iota^\ast \omega^\kj$ is a $J$-connection on $\varQ$, and $\homega_\varQ^\kk = \iota^\ast \homega^\kk$ is a $\kk$-valued $K$-gauge invariant $1$-form.
\end{proof}

Since the $J$-equivariance of $S$ is not specified, the fields $\homega_\varQ^\kk$ are not necessarily $J$-invariant. This corollary shows that the application of the Main Lemma in that specific situation corresponds to a reduction of the  principal fiber bundle $\varP \rightarrow \varQ$, which splits the connection $\omega$ into a connection on $\varQ$ and a $K$-gauge invariant $1$-form. This result is the same as the one obtained in \cite[Section~5.13]{Ster94a}, which is based, from the very beginning, on a procedure of principal fiber bundle reduction. The electroweak part of the Standard Model of particle physics is an example of this situation.

\section{Applications to gauge theories in particle physics}
\label{sec-applications-to-gauge-theories-in-particles-physics}

\subsection{The electroweak part of the Standard Model}
\label{sec-the-electroweak-part-of-the-standard-model-of-particles-physics}

This example is a simplified version of \cite{MassWall10a}, to which we refer for further details. See also \cite{Fadd09a} and \cite{CherFaddNiem08a}, where only the bosonic part of the Standard Model is considered.

The electroweak part of the Standard Model is a gauge theory with structure group $G = U(1) \times SU(2)$. In the following, the theory is presented for a trivial $G$-principal fiber bundle $\varP$ over the space-time $\varM$, so that the gauge group $\caG$ of the theory identifies with $\underline{G} = \underline{U(1)} \times \underline{SU(2)}$, which is a notation for the smooth maps from $\varM$ to $G$. All the maps and forms are then defined on $\varM$.

The part of the Lagrangian of the Standard Model we will consider is given by
\begin{equation}
\label{eq-lagrangianbeginning}
\caL[a_\mu, b_\mu, \varphi] 
= (D_\mu \varphi)^\dagger (D^\mu \varphi) - \mu^2 \varphi^\dagger \varphi - \lambda (\varphi^\dagger \varphi)^2
-\tfrac{1}{4} f_{\mu \nu} f^{\mu \nu} -\tfrac{1}{4} \sum_a g^a_{\mu \nu} g^{a\; \mu \nu}\,.
\end{equation}
In this Lagrangian, $\varphi$ is a $\gC^2$-valued scalar field\footnote{This field should be called the ``Englert–Brout–Higgs–Guralnik–Hagen–Kibble field'', from the names of the authors who discovered its importance in particle physics, while the field which gives rise to the Higgs boson after symmetry breaking deserves the name ``Higgs field''.}, $D_\mu \varphi = (\partial_\mu - i \tfrac{g}{2} b_\mu - i \tfrac{g'}{2} a_\mu) \varphi$ where $g$ and $g'$ are the coupling constants of $SU(2)$ and $U(1)$ respectively, $f_{\mu \nu}$ is the field strength of the $U(1)$-connection $a_\mu$, and $g_{\mu \nu}$ is the field strength of the $SU(2)$-connection $b_\mu$, written as $g_{\mu \nu} = g^a_{\mu \nu}\frac{\sigma_a}{2}$, where $\sigma^a$, $a=1,2,3$, are the Pauli matrices. This theory is invariant with respect to the following gauge transformations:
\begin{align*}
a_\mu^\zeta &=a_\mu + \tfrac{2i }{g'}\zeta^{-1}\partial_\mu \zeta,
&
b_\mu^\zeta &= b_\mu,
&
\varphi^\zeta &= \zeta^{-1}\varphi,
\\
a_\mu^\gamma &= a_\mu,
&
b_\mu^\gamma &=\gamma^{-1}b_\mu \gamma + \tfrac{2i}{g} \gamma^{-1}\partial_\mu \gamma,
&
\varphi^\gamma &= \gamma^{-1}\varphi,
\end{align*}
for any $\zeta \in \underline{U(1)}$ and any $\gamma \in \underline{SU(2)}$.

The scalar field $\varphi$ can be uniquely decomposed with respect to a fixed unit vector $\dotvarphi = \spmatrix{0 \\ 1}$ as $\varphi = \eta\, u \, \dotvarphi$ where $\eta : \varM \rightarrow \gR_+$ is the length of $\varphi$, and $u : \varM \rightarrow SU(2)$. With $\varphi = \spmatrix{\varphi_1 \\ \varphi_2}$, one has $\eta = \sqrt{|\varphi_1|^2 + |\varphi_2|^2}$ and $u = \tfrac{1}{\eta} \spmatrix{\overline{\varphi_2} & \varphi_1 \\ - \overline{\varphi_1} & \varphi_2}$, so that $u(x)$ is only defined when $\eta(x) \neq 0$ (more on this later). The reference vector $\dotvarphi$ defines a change of coordinates $\varphi \mapsto (\eta, u)$ in the fields space. For any $\gamma \in \underline{SU(2)}$, one has $\varphi^\gamma \mapsto (\eta^\gamma, u^\gamma)$ with $\eta^\gamma=\eta$ and $u^\gamma = \gamma^{-1} u$. Notice that for any $\zeta \in  \underline{U(1)}$, one has $\varphi^\zeta \mapsto (\eta^\zeta, u^\zeta)$ with $\eta^\zeta =\eta$ and $u^\zeta = u \hzeta$, where $\hzeta = \spmatrix{\zeta & 0 \\ 0 & \zeta^{-1}}$. In our scheme, $u$ is the \dressfield extracted from the \auxfield $\varphi$, while $\eta$ is a \resfield.

Applying the Main Lemma with the \dressfield $u$ and the connection $b_\mu$, the composite fields $B_\mu \vc u^{-1} b_\mu u + \tfrac{2i}{g} u^{-1} \partial_\mu u$ are $SU(2)$-gauge invariant. Only $SU(2)$-gauge transformations can be dealt with through the Main Lemma because $u$ does not carry the convenient representation with respect to $U(1)$-gauge transformations. The fields $B_\mu = B_\mu^a \sigma_a$ are no more $U(1)$-invariant, because $u^\zeta \neq u$. A convenient way to deal with $U(1)$-charged fields is to define $W^\pm_\mu \vc \tfrac{1}{\sqrt{2}}(B^1_\mu \mp i B^2_\mu)$, which satisfy $(W^\pm_\mu)^\zeta = \zeta^{\mp 2} W^\pm_\mu$, and to define $Z_\mu \vc \cos \theta_W B^3_\mu - \sin \theta_W a_\mu$, with $\cos \theta_W \vc \tfrac{g}{\sqrt{g^2 + g'^2}}$ and $\sin \theta_W \vc \tfrac{g'}{\sqrt{g^2 + g'^2}}$, which satisfy $Z_\mu^\zeta = Z_\mu$, so that  the fields $Z_\mu$ are invariant for the whole gauge group. Then it is natural to define $A_\mu \vc \sin \theta_W B^3_\mu + \cos \theta_W a_\mu$, which is then a $U(1)$-connection, $A_\mu^\zeta = A_\mu + 2 i \frac{1}{e} \zeta^{-1} \partial_\mu \zeta$, for the charge $e=g \sin \theta_W$. In \cite{MassWall10a}, the spinor fields of the theory give rise also to $SU(2)$-gauge invariant composite fields, as expected by the second part of the Main Lemma. 

We can now perform the two changes of variables $(a_\mu, b_\mu, \varphi) \mapsto (a_\mu, B_\mu, \eta, u) \mapsto (A_\mu, Z_\mu, W^\pm_\mu, \eta, u)$ in the Lagrangian. Thanks to the $SU(2)$-gauge invariance of the Lagrangian, the $u$ field disappears, so that $\caL[a_\mu, b_\mu, \varphi] = \caL[a_\mu, B_\mu, \eta] = \caL'[ A_\mu, Z_\mu, W^\pm_\mu, \eta]$. The Lagrangian $\caL'[ A_\mu, Z_\mu, W^\pm_\mu, \eta]$ is trivially $SU(2)$-gauge invariant because \emph{all the fields in the Lagrangian are $SU(2)$-gauge invariant}, among them the composite fields associated to the curvature and covariant derivative of $b_\mu$ described in the Main Lemma. This is almost the Lagrangian describing the electroweak part of the Standard Model after symmetry breaking in the so-called unitary gauge, which is only fully recovered after expanding $\eta$ around the \emph{unique} minimum constant configuration\footnote{Recall that $\eta >0$, so that this minimum constant configuration is unique.} $\eta_0$ of its potential $V(\eta) \vc \mu^2 \eta^2 + \lambda \eta^4$ when $\mu^2 < 0$. This expansion corresponds to introducing the field of the Higgs boson of the Standard Model.

Notice the following important facts about this procedure which has already been detailled in \cite{MassWall10a}. Firstly, it is not a symmetry breaking since the variables $W^\pm_\mu$, $Z_\mu$ and $A_\mu$ can be defined through a change of variables without making reference to any energy scale. Secondly, this change of variables induces a extra factor $\eta^3$ in the functional measure of the corresponding functional integral. Finally, notice that the ordinary symmetry breaking mechanism performs two ``transformations'' of the Lagrangian at the same time, while they are clearly distinct in the present procedure: on the one hand, some redistribution of degrees of freedom, whose similarity with Goldstone mechanism is discussed in section~\ref{sec-comments}, and, on the other hand, the creation of true mass terms. The Main Lemma corresponds to the first transformation, while the second corresponds to choosing a constant classical configuration for the field $\eta$.

Let us consider now the situation where $\varP$ is not necessarily trivial. The change of variables can only be performed at points $x \in \varM$ where $\varphi(x) \neq 0$. Let us suppose that the field $u$ exists everywhere. The field $\varphi$ is a section of the associated vector bundle $\varP \times_\ell \gC^2$, where, for any $(\zeta,\gamma) \in U(1) \times SU(2)$ and $v \in \gC^2$, $\ell_{(\zeta,\gamma)} v = \zeta \gamma v$. Such a section can be described as an equivariant map $\tvarphi : \varP \rightarrow \gC^2$. The global existence of $u$ is equivalent to the non-vanishing of $\teta(p) \vc \norm{\tvarphi(p)}_{\gC^2}$ on $\varP$. Then, one can define $S : \varP \rightarrow SU(2)$ such that $\tvarphi(p) = \teta(p) S(p) \spmatrix{0 \\ 1}$. Using the uniqueness of this decomposition, one has $S(p\gamma) = \gamma^{-1} S(p)$ for any $\gamma \in SU(2)$ and $p\in \varP$. The map $S$ is as in Prop.~\ref{prop-S-function-topology}, so that the topology of $\varP$ is trivial in the $SU(2)$ direction. The present situation corresponds to $J = U(1)$ and $K = SU(2)$ in Corollary~\ref{cor-case-J-times-K}. The composite field $\homega \vc S^{-1} \omega S + S^{-1} \dd S$ has local components $A_\mu$, $W^{\pm}_\mu$ and $Z_\mu$,  where $\omega$ has local components $a_\mu + b_\mu$. The remaining non-trivial geometry is the one of the $U(1)$-principal fiber bundle $\varP/SU(2)$, which corresponds to electromagnetism.

\subsection{$SU(N)$-gauge theories}
\label{sec-su(n)-gauge-theories}

It is natural to ask whether it is possible to apply a similar procedure to any $SU(N)$-gauge theories, with arbitrary $N$. Recall that the $U(1)$ case has been successfully treated in section~\ref{sec-general-scheme}.


For $SU(N)$ with $N>2$, it is not possible to obtain the same result since the transformation of the \dressfield $u^\gamma = \gamma^{-1}u$ is not always true. Indeed, consider the decomposition $\gC^N\ni\varphi = \eta u \dotvarphi$, with $\eta : \varM \rightarrow \gR_+$, $u : \varM \rightarrow SU(N)$, and $\dotvarphi = \spmatrix{\mathbf{0}_{N-1} \\ 1}$, where $\mathbf{0}_{N-1}$ is the zero vector in $\gC^{N-1}$. The \dressfield $u = u(\varphi)$ is only defined \emph{modulo} the right multiplication of a field $V$ of the form $\spmatrix{v&0\\0&1}$ where $v : \varM \rightarrow SU(n-1)$ acts only on $\mathbf{0}_{N-1}$.  Denote by $[u] : \varM \rightarrow SU(N)/SU(N-1)$ the composite of $u$ with the quotient map $SU(N) \rightarrow SU(N)/SU(N-1)$. Then, the reference vector $\dotvarphi$ defines a change of coordinates $\varphi \mapsto (\eta, [u])$. This implies that the gauge transformation $u^\gamma = \gamma^{-1}u$ must be substituted by the formula $[u^\gamma]=[\gamma^{-1} u]$, so that the Main Lemma cannot be applied.

Nevertheless, applying the same ideas as in the Main Lemma leads to exhibiting some interesting structures. Let $\dotvarphi \in \gC^N$ be as before, and let us choose a map $u : \varphi \mapsto u(\varphi)$ such that $\varphi = \eta u(\varphi) \dotvarphi$. In the spirit of the Main Lemma, we can define the composite field $B_\mu \vc u^{-1} b_\mu u + \tfrac{2i}{g} u^{-1} \partial_\mu u$. For any $\gamma \in \underline{SU(N)}$, we use the notation $u^\gamma \vc u(\varphi^\gamma)$. The field $V(\varphi, u,\gamma) \vc u(\varphi)^{-1} \gamma u^\gamma$ is well defined and is necessarily of the form $V= \spmatrix{v & 0 \\ 0 & 1}$ for a $v = v(\varphi, u,\gamma) : \varM \rightarrow SU(N-1)$. Then the composite field $B_\mu$ transforms as $B_\mu^\gamma \vc (u^\gamma)^{-1} b_\mu^\gamma u^\gamma + \tfrac{2i}{g} (u^\gamma)^{-1} \partial_\mu u^\gamma = V^{-1} B_\mu V + \tfrac{2i}{g} V^{-1} \partial_\mu V$. It is then convenient to decompose $B_\mu$ as $B_\mu = \spmatrix{Y_\mu + c_N B^0_\mu & \sqrt{2} X^+_\mu \\ \sqrt{2} X^-_\mu & - c'_N B^0_\mu}$ with $c_N \vc \sqrt{\frac{2}{N(N-1)}}$, $c'_N \vc \sqrt{\frac{2(N-1)}{N}}$, $X^+_\mu$  (resp. $X^-_\mu$ ) some fields with values in $\gC^{N-1}$ as a column vector (resp. as a line vector), and $Y_\mu$ some fields with values in the Lie algebra $\ksu(N-1)$. Then the induced gauge transformations on these new variables are 
\begin{align}
\label{eq-gaugeactionSU(N-1)}
Y_\mu &\mapsto v^{-1} Y_\mu v + \tfrac{2i}{g} v^{-1} \partial_\mu v,
&
X^+_\mu &\mapsto v^{-1} X^+_\mu,
&
X^-_\mu &\mapsto X^-_\mu v,
&
B_\mu^0 \mapsto B_\mu^0.
\end{align}
This implies that under an active gauge transformation $\gamma \in \underline{SU(N)}$, the $B_\mu^0$ are invariant, the $X^\pm_\mu$ are $SU(N-1)$-charged fields, and the $Y_\mu$'s behave as connection fields for a $SU(N-1)$-gauge theory.

The Lagrangian of the theory can then be written in terms of the composite fields only, because $u$ does not appears explicitly, thanks to $SU(N)$-gauge invariance. The action of the gauge group $\underline{SU(N)}$ is still implemented, but it factorizes through the map $\underline{SU(N)} \rightarrow \underline{SU(N-1)}$ given by $\gamma \mapsto v(\varphi, u,\gamma)$, which is not a morphism of groups. Formally, the Lagrangian after the change of variables describes a $SU(N-1)$-gauge theory for the gauge actions \eqref{eq-gaugeactionSU(N-1)}, where only the fields $B_\mu^0$ and $X^\pm_\mu$ can be massive, and where $Y_\mu$ are (massless)  $SU(N-1)$-connections. 

A change in the choice of the map $u : \varphi \mapsto u(\varphi)$ corresponds to introducing a new $SU(N-1)$-valued field $w$ such that $u' = u \spmatrix{w & 0 \\ 0 & 1}$. Then, the relations between the corresponding composite fields $(B^0_\mu, X^\pm_\mu, Y_\mu)$ and $(B'^0_\mu, X'^\pm_\mu, Y'_\mu)$, are of the form \eqref{eq-gaugeactionSU(N-1)}, where $v$ is replaced by $w$. In this sense, $w$ implements a gauge transformation in the $SU(N-1)$-gauge theory.

In this construction, the Standard Model, which corresponds to $N=2$, is special in the sense that the field $u$ is uniquely determined by $\varphi$ through the requirement $\varphi = \eta u(\varphi) \dotvarphi$, and one can check that $V(u,\varphi,\gamma)= e$, so that the Main Lemma can be directly applied. As one adds an extra $U(1)$ symmetry, the $B_\mu^0$ fields are the fields entering in the definition of $Z^0_\mu$, the $X_\mu^\pm$ fields are the $W_\mu^\pm$'s, and there is no $Y_\mu$ fields.

\subsection{An alternative \dressfield}
\label{sec-an-alternative-dressing-field-particles}

The definition we have taken for the \dressfield $u$ makes apparent the \resfield $\eta$ from the beginning. We propose here an alternative \dressfield, denoted by $\tu$, which gives rise to the same Lagrangian after the change of field variables.

It is well-known that, given a reference vector $\dotvarphi = \spmatrix{0\\1}$,  $\gC^2 \backslash \{0\}$ identifies as a $SU(2)$-space to $SU(2) \times \gR_+^\ast$, where $SU(2)$ acts by left multiplication on itself on the latter space. To any $\varphi \in \gC^2 \backslash \{0\}$, we then associate the unique $\tu \in SU(2) \times \gR_+^\ast$ such that $\varphi = \tu \dotvarphi$. The space $G=SU(2) \times \gR_+^\ast$ is a group which contains $H=SU(2)$ as a subgroup, and $\tu = (u,\eta)$ can be used as a \dressfield, because the composite fields $\tB_\mu \vc \tu^{-1} b_\mu \tu + \tfrac{2i}{g} \tu^{-1} \partial_\mu \tu$ are mathematically well-defined as components of a $1$-form with values in $\ksu(2) \oplus i\gR$. $\tB_\mu$ is related to the composite field defined in \ref{sec-the-electroweak-part-of-the-standard-model-of-particles-physics} by $\tB_\mu = B_\mu + i\frac{2}{g} c_\mu$, with $c_\mu=\eta^{-1}\partial_\mu\eta \in \underline \gR$. Accordingly, this choice of \dressfield is equivalent to change $B_\mu^3\in\gR$ into $B_\mu^{3}+i\frac{2}{g}c_\mu\in\gC$, whereas the components $B_\mu^{1,2}$ are the same.

A straightforward computation shows that 
\begin{equation*}
D_\mu\varphi = \eta u \spmatrix{-i\frac{g}{2} W_\mu^+ \\ c_\mu+i(\frac{g}{2} B_\mu^3  - \frac{g'}{2} a_\mu)} \in \gC^2.
\end{equation*}
 The Lagrangian we then obtain is the same as the one expressed in the variables $\eta$, $W_\mu^\pm$, $Z_ \mu$, and $A_\mu$ in \ref{sec-the-electroweak-part-of-the-standard-model-of-particles-physics}. But the main feature of this variant is that the kinetic terms for $\eta$ does not emerge from the term $\partial_\mu \varphi$ in the Lagrangian, but from the  real field $c_ \mu$ directly. 

Then, the status of $\eta$ as an observable \resfield does not depend on the precise definition of the \dressfield, and it is determinated by the Lagrangian. This shows in particular the robustness of our scheme in this case.

\section{Application to general relativity as a gauge theory}
\label{sec-application-to-general-relativity-as-a-gauge-theory}

General relativity (GR) can be described as a gauge theory on a principal fiber bundle. But, contrary to Yang-Mills fields theories, GR cannot be defined with only ordinary connections (Ehresmann connections). It requires an additional structure, namely a soldering form. We choose to encode this larger structure into a Cartan connection. We refer to \cite{Shar97a} for details on Cartan connections and its relations to Ehresmann connections.

\subsection{Geometry of Cartan connections}
\label{sec-geometry-of-cartan-connections}

We consider the Lie group $H = SO(1,m-1)$ and its Lie algebra $\kh = \kso(1,m-1)$. Let $(\varP, \varpi)$ be a reductive Cartan geometry modeled  on the Lie algebras $(\kg, \kh)$ with structure group $H$, where $\kg = \kh \oplus \gR^m$ is a $H$-module decomposition. In this setting, $\varP$ is a $H$-principal fiber bundle over space-time $\varM$ of dimension $m$, and $\varpi$ is a Cartan connection over $\varP$, which satisfies, by definition:
\begin{enumerate}[itemsep=2pt,parsep=0pt]
\item $\varpi$ is a $\kg$-valued $1$-form on $\varP$;
\item $\upR_h^\ast \varpi = \Ad_{h^{-1}} \varpi$ for any $h \in H$ where $\upR$ is the right action of $H$ on $\varP$;
\item $\varpi(X^\varP) = X$ for any $X \in \kh$, where $X^\varP$ is the associated vertical vector fields on $\varP$ for the right action;
\item for any $p\in \varP$, $\varpi_{|p}$ realizes an isomorphism of vector spaces $T_p \varP \rightarrow \kg = \kh\oplus \gR^m$.
\end{enumerate}

The existence of such a Cartan connection implies that the principal fiber bundle $\varP$ is a reduction of the $GL^+_m(\gR)$-principal frame bundle $L\varM$ of $\varM$ \cite[Lemma~A.2.1]{Shar97a}, but it is not a natural bundle in the sense of \cite[Section~14]{KolaMichSlov93a}. Moreover, the Cartan connection $\varpi$ decomposes into two pieces, as $\varpi = \omega \oplus \beta$, corresponding to the decomposition $\kg = \kh \oplus \gR^m$. The $1$-form $\omega$, with values in $\kh$, is an Ehresmann connection on $\varP$, and $\beta$ is a tensorial $1$-form on $\varP$, which induces a non canonical isomorphism $\Phi_p : T_x \varM \rightarrow \gR^m$ for any $x \in \varM$ and $p \in \pi^{-1}(x)$, where $\pi : \varP \rightarrow \varM$ is the projection.

The group $SO(1,m-1)$ is associated to a (fixed) metric $\eta$ on $\gR^m$, whose components we denote by $(\eta_{ab})$. For any $x \in \varM$ and $p \in \pi^{-1}(x)$, the isomorphism $\Phi_p$ induces a metric $g_x$ on $T_x \varM$ by $g_x(X_{|x}, Y_{|x}) \vc \eta (\Phi_p(X_{|x}), \Phi_p(Y_{|x}))$. In other words, $\beta$ defines a metric $g$ on the base manifold $\varM$.

Let $\Xi : \varP \rightarrow H$, with $\Xi(ph) = \Ad_{h^{-1}}\Xi(p)$, be a gauge transformation of the principal fiber bundle $\varP$. It acts on $\omega$ and $\beta$ as $\omega^\Xi = \Xi^{-1} \omega \Xi + \Xi^{-1} \dd \Xi$ and $\beta^\Xi = \Xi^{-1} \beta$. This last relation suggests that the field $\beta$ could have been a good candidate, as a \dressfield, to define a gauge invariant composite field out of $\omega$. But the composite expression \eqref{eq-compositefield} would not be mathematically well-defined because $\beta^{-1}$ does not make sense. To define a composite field, we have to look at trivializations of these fields.

\subsection{Application of the Main Lemma}
\label{sec-application-of-the-main-lemma-gravity}

Let $\varU \subset \varM$ be an open subset such that $\varP$ can be trivialized with a local section $s : \varU \rightarrow \varP$ over $\varU$, and such that there is a coordinate system $(x^\mu)$ for $\varM$ over $\varU$. We define the local $1$-forms on $\varU$: $\Gamma \vc s^\ast \omega$ and $\Lambda \vc s^\ast \beta$. Using the coordinate system, $\Gamma$ is a matrix valued $1$-form $(\Gamma^a{}_{b\mu} \dd x^\mu)$ for $a,b=1,\dots,m$, and $\Lambda$ is a $\gR^m$-valued $1$-form $(\Lambda^a{}_\mu \dd x^\mu)$. The induced isomorphism $\Phi_p$ is given by $T_x \varM \ni X = X^\mu \partial_\mu \mapsto \Lambda_{|x}(X) = \Lambda^a{}_\mu(x) X^\mu\in\gR^m$, so that the matrix-valued function $(\Lambda^a{}_\mu)$ is invertible everywhere on $\varU$. For any gauge transformation $\Xi$ on $\varP$, we define its local expression $\gamma \vc s^\ast \Xi : \varU \rightarrow H$. Then the gauge transformations take the form $\Gamma^\gamma = \gamma^{-1} \Gamma \gamma + \gamma^{-1} \dd \gamma$ and $\Lambda^\gamma = \gamma^{-1} \Lambda$. 

For any vector spaces $V,W$, denote by $L(V,W)$ the space of linear maps $V \rightarrow W$. Then, for any $x \in \varU$, one has $\Gamma_{|x} \in T^\ast_x \varM \otimes L(\gR^m, \gR^m)$, and we use the identifications $\Lambda_{|x} \in L(T_x \varM, \gR^m)$, $\Lambda^{-1}_{|x} \in L(\gR^m, T_x \varM)$, and $(\dd \Lambda)_{|x} \in T^\ast_x \varM \otimes L(T_x \varM, \gR^m)$, so that $\Lambda^{-1}_{|x} \Gamma_{|x} \Lambda_{|x} \in T^\ast_x \varM \otimes L(T_x \varM, T_x \varM)$ and $\Lambda^{-1}_{|x} (\dd \Lambda)_{|x} \in T^\ast_x \varM \otimes L(T_x \varM, T_x \varM)$, where the products are compositions of linear maps on vector spaces. Using these mathematically well-defined expressions, we can apply the Main Lemma with the \dressfield $\Lambda$ to define the composite field
\begin{equation}
\label{eq-hgamma-lambda-gamma}
\hGamma \vc \Lambda^{-1} \Gamma \Lambda + \Lambda^{-1} \dd \Lambda,
\end{equation}
which is a gauge invariant field of $1$-forms on $\varU$ with values in $L(T_x \varM, T_x \varM)$. Using the coordinate system, $\hGamma$ is a collection of fields $\hGamma^\nu{}_{\rho\mu} = {\Lambda^{-1}}^\nu{}_ a \Gamma^a{}_{b\mu} \Lambda^b{}_\rho + {\Lambda^{-1}}^\nu{}_ a \partial_\mu \Lambda^a{}_\rho$. The gauge invariance is related to the fact that the latin indices (in terms of which the action of $H$ is written) have disappeared in favor of the geometric greek indices related to the coordinate system. The field $\hGamma$ is not a tensor field with respect to a change of coordinate system, and it behaves as Christoffel symbols. It defines a linear connection on the natural geometry of $\varM$ (the geometry of tensor fields), and it satisfies the metric condition $D^{\hGamma} g=0$ inherited from the (trivially satisfied) relation $D^{\omega} \eta = 0$. This example is particularly relevant to illustrate the procedure of geometrization described in the introduction. Some authors have interpreted the defining relation \eqref{eq-hgamma-lambda-gamma} as a gauge transformation of $\Gamma$ by an element of the gauge group $GL^+_m(\gR)$. But here, \eqref{eq-hgamma-lambda-gamma} is not a gauge transformation because $\Lambda$ is not in the gauge group of the initial $SO(1,m-1)$-gauge theory, and $\hGamma$ is no more a $SO(1,m-1)$-connection. 

General relativity is encoded by the gauge invariant Lagrangian (see \textsl{e.g.} \cite{GockSchu89a})
\begin{equation}
\label{eq-lagrangian-gauge-RG}
\caL[\Gamma, \Lambda] \vc R^a{}_b \wedge \ast(\Lambda^b \wedge \Lambda_a),
\end{equation}
where $(R^a{}_b)$ is the curvature of the connection $1$-form $\Gamma$, $\ast$ is the Hodge star operator defined by the metric $g$ (induced by $\beta$), and the lowering of the latin indices is done using the metric $(\eta_{ab})$ on $\gR^m$. Performing the (invertible) change of field variables $(\Gamma, \Lambda) \mapsto (\hGamma, \Lambda)$ in the Lagrangian, one gets the usual Einstein Lagrangian
\begin{equation*}
\caL[\hGamma, \Lambda] = \sqrt{\abs{g}} \hR_\text{scalar},
\end{equation*}
where $\abs{g}$ is the determinant of the metric matrix, and $\hR_\text{scalar}$ is the scalar curvature of the Christoffel symbols $\hGamma^\nu{}_{\rho\mu}$, which is obtained from the composite field $\hR$ of the curvature $R$ of $\Gamma$. Notice that the metric $g$ enters in the Lagrangian through $\Lambda$, so that this Lagrangian depends only on $\hGamma$ and $g$. The metric $g$ is the \resfield in our scheme.

The curvature of the Cartan connection $\varpi$ contains two terms: the first one is the curvature of  $\omega$, which has been used in the Lagrangian, and the second one is the covariant derivative $D\beta$ of $\beta$ along the $\omega$. Locally on $\varU$, this covariant derivative is the $2$-form with values in $\gR^m$: $\phi^a{}_{\mu \rho} \dd x^\mu \wedge \dd x^\rho \vc (\partial_\mu \Lambda^a{}_\rho + \Gamma^a{}_{b\mu} \Lambda^b{}_\rho) \dd x^\mu \wedge \dd x^\rho$. To this field, one can associate the gauge invariant composite field $\hphi^\nu \vc {\Lambda^{-1}}^\nu{}_a \phi^a{}_{\mu \rho} \dd x^\mu \wedge \dd x^\rho = \hGamma^\nu{}_{\rho\mu} \dd x^\mu \wedge \dd x^\rho$, which is the torsion $2$-form associated to the Christoffel symbols $\hGamma^\nu{}_{\rho\mu}$.

The change of field variables presented here gives rise to the same usual computations which relate the degrees of freedom of the gauge formulation of GR to its original geometrical formulation. The usual point of view consists in using the metric $g$ to perform a symmetry reduction of the $GL^+_m(\gR)$-principal frame bundle to the subgroup $SO(1,m-1)$. Instead of that, our procedure performs a reduction of the gauge symmetry group $SO(1,m-1)$ to ``nothing'', so that  we end up with a geometrical theory in the sense of section~\ref{sec-introduction}: in our point of view, the $GL^+_m(\gR)$-principal frame bundle belongs to the natural geometry of $\varM$.

\subsection{An alternative \dressfield}
\label{sec-an-alternative-dressing-field-gravity}

The \dressfield $\Lambda$ looks very much like the \dressfield $\tu$ defined in \ref{sec-an-alternative-dressing-field-particles}, in the sense that it contains at the same time the degrees of freedom of the gauge group and those of the \resfield $g$. It is possible to define an alternative \dressfield which contains only the gauge group degrees of freedom and makes apparent the \resfield from the beginning.

The procedure consists in writing a decomposition of $\Lambda$ in terms of two fields: one containing the $\frac{m(m-1)}{2}$ degrees of freedom of $SO(1,m-1)$, and the second one containing the remaining $\frac{m(m+1)}{2}$ degrees of freedom. From \ref{sec-application-of-the-main-lemma-gravity}, we know that $(\Lambda^a{}_\mu)$ is an invertible matrix, so that removing the degrees of freedom of $SO(1,m-1)$ from $\Lambda$ amounts to consider the quotient $GL^+_m(\gR) / SO(1,m-1)$, which, at the level of the $GL^+_m(\gR)$-principal frame bundle $L\varM$, corresponds to the choice of a metric $g$ on $\varM$ (with $\frac{m(m+1)}{2}$ degrees of freedom). But performing this quotient does not explicitly separate $\Lambda$ into a \dressfield and a \resfield.

Another way to proceed is to consider the matrix $(\Lambda^a{}_\mu)$ as a set of $m$ vectors $v_\mu =  (\Lambda^a{}_\mu)$ which defines a basis of $\gR^m$. The extraction of the \dressfield then consists to orthonormalize this basis for the metric $(\eta_{ab})$. Because the signature of $\eta$ is not Euclidean, the usual Gram–Schmidt process does not work. 

We rely on a procedure which works for any signature $(r,s)$, $r+s=m$ \cite{SimoChatSrin99a,BlytRobe02a}. Denote by $I_{r,s}$ the matrix $\spmatrix{\Id_r & 0 \\ 0 & -\Id_s}$, which represents the metric $\eta$ on $\gR^m$, and let $G = (\eta(v_\mu, v_\nu))$ be the Gram matrix of the basis $\{v_\mu\}$. It is a real symmetric matrix which can be diagonalized in the form $G = R D R^T$ where $R \in SO(m)$, $D = \diag(\lambda_1, \dots, \lambda_m)$, and $R^T$ is the transpose matrix. $G$ is the metric $g$, thus it is non degenerate, and it has the signature $(r,s)$: all the real numbers $\lambda_i$ are non zero, and we can choose an ordering of the $\lambda_i$'s such that $\lambda_i >0$ for $i \leq r$ and $\lambda_i < 0$ for $i > r$. Consider now  $U \vc \Lambda R \abs{D}^{-1/2}$ where $\abs{D} \vc \diag(\abs{\lambda_1}, \dots, \abs{\lambda_m})$. Then a straightforward computation shows that $U^T I_{r,s} U = I_{r,s}$, using the facts that $G = \Lambda^T I_{r,s} \Lambda$ and $\abs{D}^{-1/2} D \abs{D}^{-1/2} = I_{r,s}$. The matrix $U \in SO(r,s)$ then defines a pseudo-orthonormal basis on $(\gR^m, \eta)$.

Let us define the local field $T \vc U^{-1} \Lambda = \abs{D}^{1/2} R^{-1}$. The matrices $G$, $R$ and $D$ are gauge invariant, so that $T$ is gauge invariant. On the other hand, a gauge transformation $\Lambda \mapsto \gamma^{-1} \Lambda$ induces a transformation $U \mapsto \gamma^{-1} U$. The local field $U$ can then be considered as a \dressfield extracted from $\Lambda^a{}_\mu = U^a{}_b T^b{}_\mu$. Applying the Main Lemma, we define the gauge invariant local field $\bGamma \vc U^{-1} \Gamma U + U^{-1} \dd U$. Performing this change of field variables in \eqref{eq-lagrangian-gauge-RG}, the Lagrangian can be locally written in terms of gauge invariant fields as:
\begin{equation*}
\caL[\bGamma, T] = \det(T) \, \eta^{bc}\, {T^{-1}}^{\nu}{}_a \bR^a{}_{b,\mu\nu} {T^{-1}}^{\mu}{}_c
\end{equation*}
where $\bR \vc \dd \bGamma + \frac{1}{2} [\bGamma, \bGamma] = (\bR^a{}_{b,\mu\nu} \dd x^{\mu} \wedge \dd x^{\nu})$ is the ``curvature'' of $\bGamma$. The field $T$ is then the natural \resfield in this formulation. The field $R$ and  $\abs{D}^{1/2}$ contain respectively $\frac{m(m-1)}{2}$ and $m$ degrees of freedom, so that $T$ contains exactly the $\frac{m(m+1)}{2}$ degrees of freedom of $\Lambda$ entering in the metric $g_{\mu\nu} \vc \eta_{ab} \Lambda^a{}_\mu \Lambda^b{}_\nu = \eta_{ab} T^a{}_\mu T^b{}_\nu$.

Given the Gram matrix $G$, the matrices $R$ and $D$ such that $G = R D R^T$ are not unique. Firstly, one can permute any couple of eigenvalues $\lambda_i$ and $\lambda_j$ provided that they have the same sign as required by the construction. This implies that for any $P \in \kS_r \times \kS_s \subset \kS_m$, represented as a $m\times m$ permutation matrix, $R' = RP$ and $D' = P^T D P$ is also a possible choice. Secondly, when $\lambda_i = \lambda_j$, for any rotation $S$ in the corresponding eigenspace of this eigenvalue, $R' = RS$ and $D' = S^T D S = D$ is also a possible choice. In both situation, one has $P, S \in SO(r)\times SO(s)$.

The above method applies to the Euclidean case as well: it is known as the Schweinler-Wigner orthogonalisation procedure, and $U$ is the Schweinler-Wigner basis \cite{SchwWign70a, ChatKapoSrin98a}. Notice that the Gram-Schmidt orthogonalization procedure would have provided us with a unique decomposition $\Lambda = Q R$, where $Q \in SO(m)$ represents an orthonormal basis, and $R$ is upper triangular with positive entries on the diagonal. This is the so called $QR$ decomposition. Both $Q$ and $U$ can be used as \dressfields for Euclidean gravity. For signature $(r,s)$, the pseudo-orthonormal basis $U$ could be named ``Schweinler-Wigner basis'' since it is constructed following the same scheme and displays analogous properties \cite{SimoChatSrin99a}. 

The decomposition $\Lambda = UT$ is defined on an open subset $\varU \subset \varM$ over which the $SO(r,s)$-principal fiber bundle is trivialized. Let us consider two such open subsets $\varU_i, \varU_j$ such that $\varU_i \cap \varU_j \neq \varnothing$. Then, with obvious notations, there exists $h : \varU_i \cap \varU_j \to SO(r,s)$ such that $\Lambda_j = h^{-1} \Lambda_i$ and $\Gamma_j = h^{-1} \Gamma_i h + h^{-1} \dd h$. The Gram matrices define a global structure $G$ (the metric $g$) on $\varM$ since $G_j = \Lambda_j^T I_{r,s} \Lambda_j = \Lambda_i^T h I_{r,s} h^T \Lambda_i = G_i$, so that the matrices $R_i, R_j, D_i, D_j$ are related over $\varU_i \cap \varU_j$ by a $SO(r)\times SO(s)$-matrix valued function $P$ such that $D_j = P^T D_i P$ and $R_j = R_i P$. This implies that, over $\varU_i \cap \varU_j$, one has $T_j = P^T T_i$ and $U_j = h^{-1} U_i P$. A straightforward computation then gives $\bGamma_j = P^{-1} \bGamma_i P + P^{-1} \dd P$. In this relation, $P$ does not depend on the geometry of the principal fiber bundle, so that the $\bGamma_i$'s do not depend on this fiber bundle. Nevertheless, $P$ depends on some choices performed over each $\varU_i$, and the $R_i$'s and $D_i$'s don't have convenient transformations under coordinate changes on $\varM$, thus the $\bGamma_i$'s are not defined globally on $\varM$.

The \dressfield $U$ then induces local $1$-forms $\bGamma$ independent of the $SO(r,s)$-principal fiber bundle, but contrary to the \dressfield $\Lambda$ used in \ref{sec-application-of-the-main-lemma-gravity}, it does not produce global geometrical objects on top of $\varM$. Finally, notice that the decomposition $G = R D R^T$ is ill-defined at points $x \in \varU$ such that $\lambda_i(x) = \lambda_j(x)$ if one requires $R$ to be smooth. This decomposition can only be used for generic situations.

\section{Application to a Yang-Mills-Higgs model based on Atiyah Lie algebroids}
\label{sec-application-to-a-yang-mill-higgs-model-based-on-atiyah-lie-algebroids}

The notion of connections admits many generalizations. Among them, some of us have developed the necessary mathematical structures which permit to define gauge field theory from transitive Lie algebroids. As a particular transitive Lie algebroid, we will use, in this example, the  Atiyah Lie algebroid of a $G$-principal fiber bundle $\varP$, for a connected Lie group $G$. We will use notations and results from \cite{LazzMass12a} and \cite{FourLazzMass13a}. The general theory of Lie algebroids can be found in \cite{Mack05a}.

\subsection{Generalized connections on Atiyah Lie algebroids}
\label{sec-generalized-connections-on-atiyah-lie-algebroids}

The transitive Atiyah Lie algebroid of a $G$-principal fiber bundle $\varP$ is defined as the short exact sequence of Lie algebras and $C^\infty(\varM)$-modules
\begin{equation*}
\xymatrix@1{{\algzero} \ar[r] & {\Gamma_G(\varP, \kg)} \ar[r]^-{\iota} & {\Gamma_G(T\varP)} \ar[r]^-{\pi_\ast} & {\Gamma(T \varM)} \ar[r] & {\algzero}},
\end{equation*}
with
\begin{align*}
\Gamma_G(T\varP) &= \{ \sfX \in \Gamma(T\varP) \mid \upR_{g\,\ast}\sfX = \sfX \text{ for all } g \in G \},
\\
\Gamma_G(\varP, \kg) &= \{ v : P \rightarrow \kg \mid v(p g) = \Ad_{g^{-1}} v(p) \text{ for all } g \in G \},
\end{align*}
where $\upR$ denotes the right action of $G$ on $\varP$ and $\iota$ is given by $\iota(v)(p) = \left( \frac{d}{dt} p  e^{-t v(p)} \right)_{|t=0}$. We will use the short notation $\lieA \vc \Gamma_G(T\varP)$ for the Lie algebroid, and $\lieL \vc \Gamma_G(\varP, \kg)$ for its kernel.

There is a natural notion of forms on $\lieA$ with values in the kernel $\lieL$, which defines a graded differential Lie algebra $(\Omega^\grast(\lieA, \lieL), \hd)$, where $\hd$ extends the de~Rham differential by a purely algebraic differential operator.  A generalized connection on $\lieA$ is defined to be a $1$-form $\varpi \in \Omega^1(\lieA, \lieL)$, and its curvature is defined as the $2$-form $R \vc \hd \varpi + \frac{1}{2}[\varpi, \varpi] \in \Omega^2(\lieA, \lieL)$. 

Let $\varE \vc \varP \times_\ell F$ be an associated vector bundle to $\varP$ for the representation $\ell$ of $G$ on a vector space $F$. We denote by $\ell_\ast$ the induced representation of $\kg$ on $F$. We identify the space of smooth sections as $\Gamma(\varE) = \{ \phi : \varP \rightarrow F \mid \phi(pg) = \ell(g^{-1}) \phi(p) \}$. Any connection $\varpi$ on $\lieA$ induces a covariant derivative $\lieA \ni \sfX \mapsto \hnabla_\sfX$ on $\Gamma(\varE)$ by the relation $\hnabla_\sfX \phi \vc \sfX\cdotaction \phi + \ell_\ast (\varpi(\sfX)) \phi$. 

We denote by $\caG$ the gauge group of $\varP$. An element $\Xi \in \caG$ is a map $\Xi : \varP \rightarrow G$ with $\Xi(pg) = \Ad_{g^{-1}}\Xi(p)$, and it acts naturally on $\Gamma(\varE)$: $\phi \mapsto \phi^\Xi \vc \ell(\Xi^{-1})\phi$. This action induces a natural action on the space of generalized connections through the requirement $ \hnabla^\Xi_\sfX \phi^\Xi = (\hnabla_\sfX \phi)^\Xi \vc \ell(\Xi^{-1}) \hnabla_\sfX \phi$ for any $\phi \in \Gamma(\varE)$ and any $\sfX \in \lieA$. Explicitly, one has $\varpi^\Xi = \Ad_{\Xi^{-1}} \varpi + \Xi^{-1} \hd \Xi$, where $\Xi^{-1} \hd \Xi \in \Omega^1(\lieA, \lieL)$ is defined as $\sfX \mapsto \Xi^{-1} (\sfX\cdotaction \Xi) \in \lieL$.

Ordinary connections on $\varP$ are contained in this space of generalized connections \cite{LazzMass12a}: a generalized connection $\varpi \in \Omega^1(\lieA, \lieL)$ is an ordinary connection if, and only if, $\varpi \circ \iota = -\Id_\lieL$. This inclusion is compatible with the respective notions of curvature and gauge group actions. In particular, the space of ordinary connections is stable under the action of the gauge group. To any generalized connection $\varpi$ on $\lieA$, we associate its reduced kernel endomorphism $\tau : \lieL \rightarrow \lieL$ defined by $\tau \vc \varpi \circ \iota + \Id_\lieL$. This endomorphism on $\lieL$ is the obstruction for $\varpi$ to be an ordinary connection. Under the previously defined gauge transformations, one has $\tau^\Xi = \Ad_{\Xi^{-1}} \tau$. The reduced kernel endomorphism associated to a generalized connection represents the ``algebraic'' part of the connection. In order to extract the geometric part, it is necessary to introduce a fixed background ordinary connection $\dotomega$ on $\varP$ (see \cite{FourLazzMass13a} for details). Then $\omega \vc \varpi + \tau \circ \dotomega$ is an ordinary connection, \textsl{i.e.} it satisfies $\omega \circ \iota = - \Id_\lieL$, and it transforms as a connection under gauge transformations: $\omega^\Xi = \Xi^{-1} \omega \Xi + \Xi^{-1} \hd \Xi$.

\subsection{Application of the Main Lemma}
\label{sec-application-of-the-main-lemma-algebroids}

In the following, we consider a gauge theory based on a subspace of the total space of generalized connections, which we require to be stable by gauge transformations. A minimal and convenient candidate consists in fixing an element $\dottau$ in the space of endomorphisms of $\lieL$ (as sections of a vector bundle), and in collecting all the generalized connections whose reduced  kernel endomorphism is of the form $\Ad_{u} \dottau$ for any $u \in \caG$. We denote this subspace of generalized connections by $\caA_{\dottau}$. Ordinary Yang-Mills theories correspond to the choice $\caA_0$ (\textsl{i.e.} $\dottau = 0$), and, in the following, the choice $\caA_{\Id_\lieL}$ (\textsl{i.e.} $\dottau = \Id_\lieL$) will be considered. The general situation is much more involved and outside the scope of the present paper, but will be studied in a forthcoming paper. In order to further simplify our model, we assume that the group $G$ is such that $\Ad$ is faithful, or, since $G$ is connected, that $G$ is centerless. This implies that the reduced kernel endomorphism $\tau$ associated to any generalized connection in $\caA_{\Id_\lieL}$ can be parametrized by the variable $u \in \caG$ as $\tau = \Ad_{u} \dottau$ with $\dottau=\Id_\lieL$. Let $\Xi$ be a gauge transformation. Then to $\tau^\Xi$ is uniquely associated $u^\Xi \in \caG$ such that $u^\Xi \dottau (u^\Xi)^{-1} = \Xi^{-1} u \dottau u^{-1} \Xi$, so that $u^\Xi = \Xi^{-1} u$.

We apply the Main Lemma to the (ordinary) connection $\omega$ with the \dressfield $u$, and we define the composite field $\homega \vc u^{-1} \omega u + u^{-1} \hd u$. We can summarize the previous steps by the successive changes of field variables $\varpi \overset{\dotomega}{\mapsto} (\omega,\tau) \overset{\dottau}{\mapsto} (\omega,u) \mapsto (\homega, u)$. In the same way, any $\phi \in \Gamma(\varE)$ defines a gauge invariant composite field $\hphi \vc u^{-1} \phi$.

In this example, the \dressfield $u$ belongs to the $\caG$-space $(\caG, \sfR)$. Let us comment this important point. In the ordinary differential geometry of fiber bundles with connections, we have recalled in section~\ref{sec-general-scheme} that the gauge group action is usually defined through the \emph{geometrical} action of $\caG$, which moves points of $\varP$ along its fibers. This action is then induced on the various associated elements of the theory, for instance through pull-back on functions and forms on $\varP$. In the present situation, the gauge group of the theory is also the gauge group of $\varP$, but its action on the space of generalized connections $\Omega^1(\lieA, \lieL)$ \emph{is not} induced by the geometry of $\lieA$ and $\lieL$. Indeed, this action has been defined by the \emph{field theoretical} requirement that (generalized) covariant derivatives transform homogeneously.\footnote{This requirement corresponds in physics to the so-called gauge principle \cite{ORai86a}.} On the subspace of ordinary connections, these two gauge actions coincide, but they do not on the whole of $\Omega^1(\lieA, \lieL)$. Reported on the field $u$, this action is $\sfR$ but not $\alpha$.

\subsection{The model and its physical content}
\label{sec-the-model-and-its-physical-content}

In order to understand the meaning of the composite fields constructed in this example, we consider a gauge invariant Lagrangian for connections in $\caA_{\dottau}$ and ``matter fields'' $\phi \in \Gamma(\varE)$. We refer to \cite{FourLazzMass13a} for details on the  construction of a gauge invariant action functional, of which we only describe the salient results here. The construction of this action functional  requires a non degenerate and inner non degenerate metric $\hg$ on $\lieA$, which can be decomposed into three pieces, $(g,h,\dotomega)$, where $g$ is a (non-degenerate) metric on the base manifold $\varM$, $h$ is a non-degenerate metric on $\lieL$, and $\dotomega$ is an ordinary connection on $\varP$, which will be our fixed background ordinary connection. The inner metric $h$ is required to be a Killing inner metric, to ensure that the Lagrangian is gauge invariant.

Let $\varU \subset \varM$ be an open subset which trivializes $\varP$. Then $\tau$, $\omega$, and $\dotomega$, have local expressions $\tau_\loc$, $\omega_\loc \vc A - \theta$, and $\dotomega_\loc \vc  \dotA - \theta$, respectively, for $\tau_\loc \in C^\infty(\varU) \otimes \End(\kg) = C^\infty(\varU) \otimes \kg^\ast \otimes \kg$, where $A, \dotA \in \Omega^1(\varU) \otimes \kg$ are local connection $1$-forms on $\varU$, and where $\theta \in \kg^\ast\otimes \kg$ is the Maurer-Cartan form on $G$. The local expression of the curvature of $\varpi$ decomposes into three terms of specific bidegrees in $\Omega^\grast(\varU) \otimes \exter^\grast \kg^\ast \otimes \kg$.  They can be expressed using the (ordinary) curvatures $F$ and $\dotF$ of $A$ and $\dotA$. The first one, of bidegree $(2,0)$, is $\tF \vc F - \tau_\loc \circ \dotF$; the second one, of bidegree $(1,1)$, is $D\tau_\loc \vc \dd \tau_\loc + [A, \tau_\loc] - \tau_\loc([\dotA, \theta])$; and the third one, of bidegree $(0,2)$, is  $W \vc \tau_\loc([\theta, \theta]) - [\tau_\loc, \tau_\loc]$. The Lagrangian is defined using a Hodge star operation induced by the metric $\hg$, and the curvature of $\varpi$. Locally, this Lagrangian reduces to the sum of the squares of the above three terms, where the contractions are performed using $g_{\mu\nu}$ for the geometric indices and $h_{ab}$ for the indices along the Lie algebra $\kg$ in a given basis $\{E_a\}$.

Let us now return to the case $\dottau = \Id_\lieL$, but maintaining the notation $\dottau$ for a while. In \cite{FourLazzMass13a}, it is shown that $\dottau=\Id_\lieL$ is the unique element in the gauge orbit $\{\Ad_u \dottau \mid u \in \caG \}$ which is trivialized as $\theta$ in any local trivialization of $\varP$. The composite field $\homega$ has a local expression $\hA - \theta$ on $\varU$, with $\hA = u_\loc^{-1} A u_\loc + u_\loc^{-1} \dd u_\loc$, where $u_\loc : \varU \rightarrow G$ is the local expression of $u \in \caG$, while $\tau_\loc = \Ad_{u_\loc}(\dottau_\loc)$. Using the change of variables $(A, \tau_\loc) \mapsto (\hA, u_\loc)$, a direct computation shows that $\hF \vc  \dd \hA + \tfrac{1}{2}[\hA, \hA] =  \Ad_{u_\loc^{-1}}(F)$ is the composite field of the curvature $F$ of $A$ as described in the Main Lemma, so that 
\begin{align}
\tF &= \Ad_{u_\loc} (\hF - \dottau_\loc \circ \dotF),
&
D\tau_\loc &= \Ad_{u_\loc} (\hD \dottau_\loc),
&
W &= \Ad_{u_\loc} (\dotW),
\label{eq-curvature-terms}
\end{align}
where $\hD \dottau_\loc \vc \dd \dottau_\loc + [\hA, \dottau_\loc] - \dottau_\loc([\dotA, \theta])$, and $\dotW \vc \dottau_\loc([\theta, \theta]) - [\dottau_\loc, \dottau_\loc]$. Since $\dottau = \Id_\lieL$, \eqref{eq-curvature-terms} simplifies into $\tF = \Ad_{u_\loc} (\hF - \dotF)$, $D\tau_\loc = \Ad_{u_\loc}([\hA - \dotA, \theta])$, and $W=0$.

With the help of a convenient metric on the vector bundle $\varE$, we can add to the Lagrangian a term coupling $\varpi$ with a field $\phi \in \Gamma(\varE)$, using again the Hodge star operation and the covariant derivative $\hnabla \phi$. Locally, $\hnabla \phi$ decomposes into two terms, $\dd \phi + \ell_\ast(A) \phi = \ell(u_\loc) (\dd \hphi + \ell_\ast(\hA) \hphi)$, and $-\ell_\ast\tau(\dotomega_\loc) \phi = - \ell(u_\loc) (\ell_\ast\dottau (\dotomega_\loc) \hphi) = - \ell(u_\loc) \ell_\ast(\dotomega_\loc) \hphi$. The induced terms in the Lagrangian are the squares of these two terms. 

By gauge invariance, the field $u_\loc$ disappears from the Lagrangian, so that the transformed action depends only on  $\hA$ and $\hphi$, and on the triple $(g,h,\dotomega)$.  The gauge field part of the Lagrangian is the sum of the square of $[\hA - \dotA, \theta]$, which induces a mass term for the field $\hA$, and the square of $\hF - \dotF$, which is a kinetic term, \textsl{à la} Yang-Mills, for $\hA$. The matter field part of the Lagrangian contains an ordinary minimal coupling between $\hphi$ and $\hA$, and a mass term for $\hphi$ coming from the square of $\ell_\ast(\dotomega_\loc) \hphi$. 

Consequently, the original Lagrangian describes \emph{massless gauge} fields $A$, $\tau$ and $\phi$, while, after application of the Main Lemma, it describes \emph{gauge invariant massive} vector and matter fields $\hA$ and $\hphi$. Similarly to what the (usual) Higgs mechanism does through a ``phase transition'', this theory solves the problem of combining gauge symmetries with massive vector fields. The only difference with an ordinary Yang-Mills field theory coupled to massless matter fields lies in the choice of the space of ``admissible'' generalized connections: $\caA_0$ \textsl{versus} $\caA_{\Id_\lieL}$.

Contrary to the case of the Standard Model described in \ref{sec-the-electroweak-part-of-the-standard-model-of-particles-physics}, the change of variables performed in this toy model does not induce an extra factor in the functional measure of the corresponding functional integral when the group $G$ is unimodular (in the sense that $\det(\Ad_g) = 1$ for any $g \in G$). Moreover, the construction can be done for any connected centerless Lie group $G$, for instance for the family of groups $SU(N)/\gZ_N$ for any $N>1$, which are the typical groups used in particle physics \cite{ORai86a}.

\section{Comments}
\label{sec-comments}

\begin{table}[t]
\newcolumntype{C}{>{\centering\arraybackslash}m{10.7em}}
\newcolumntype{D}{>{\centering\arraybackslash}m{8.5em}}
\renewcommand{\arraystretch}{1.2}
\setlength{\tabcolsep}{0pt}
\small
\centering
\begin{tabular}{@{}DCCC@{}}
\toprule
Gauge theory
&
E-W part of the\linebreak Standard Model\linebreak (\caseone)
&
Einstein's theory\linebreak of Gravity\linebreak (\casetwo)
&
Yang-Mills-Higgs theory\linebreak
on Atiyah Lie algebroid\linebreak (\casethree)
\\
\toprule
Structure group 
&
$U(1)\times SU(2)$
&
$SO(1,m-1)$
&
$G$
\\
\textsl{dimension $(1)$}
&
$1+3$
&
$\frac{m(m-1)}2$
&
$n$
\\
\midrule
Connections $\omega$
&
$a_\mu + b_\mu$
&
$\Gamma$
&
$\omega$
\\
\textsl{dimension $(2)$}
&
$m+3m$
&
$\frac{m^2(m-1)}2$
&
$mn$
\\
\midrule
\Auxfield
&
$\varphi$
&
$\Lambda$
&
$\tau$ s.t. $\tau=u^{-1}\dottau u$
\\
\textsl{dimension $(3)$}
&
$4$
&
$m^2$
&
$n$
\\
\midrule
Reference configuration
&
$\dotvarphi= \spmatrix{0\\1}$
&
$(\dd x^\mu)$
&
$\dottau$
\\
\midrule
\Dressfield $u$
&
$u$ s.t. $\varphi=u\eta\dotvarphi$
&
$(\Lambda^a{}_\mu)$ s.t. $\beta^a = \Lambda^a{}_\mu \dd x^\mu$
&
$u$ s.t. $\tau=u^{-1}\dottau u$
\\
\midrule
Degrees of freedom
\linebreak
of the theory
\linebreak
$(2)+(3)-(1)$
&
$4m$
&
$\frac{m(m^2+1)}2$
&
$mn$
\\
\bottomrule
\end{tabular}
\caption{Fields involved in the three examples described in the text, with their meanings and degrees of freedom, before applying the Main Lemma.}
\label{table-before-MainLemma}
\end{table}

\begin{table}[t]
\newcolumntype{C}{>{\centering\arraybackslash}m{10.7em}}
\newcolumntype{D}{>{\centering\arraybackslash}m{8.5em}}
\renewcommand{\arraystretch}{1.4}
\setlength{\tabcolsep}{0pt}
\small
\centering
\begin{tabular}{@{}DCCC@{}}
\toprule
Gauge theory
&
E-W part of the\linebreak Standard Model\linebreak (\caseone)
&
Einstein's theory\linebreak of Gravity\linebreak (\casetwo)
&
Yang-Mills-Higgs theory\linebreak
on Atiyah Lie algebroid\linebreak (\casethree)
\\
\toprule
Residual group
&
$U(1)$
&
$\{e\}$
&
$\{e\}$
\\
\textsl{dimension $(1)$}
&
$1$
&
$0$
&
$0$
\\
\midrule
Gauge invariant
\linebreak
composite fields
\linebreak
$\homega=u^{-1}\omega u + u^{-1}\dd u$
&
$A_\mu$, $W_\mu^\pm$, $Z_\mu$
&
$\hGamma^{\nu}{}_{\rho\mu}$ s.t. $D^{\widehat\Gamma} g=0$
&
$\hA_\mu$
\\
\textsl{dimension $(2)$}
&
$m+3m$
&
$\frac{m^2(m-1)}2$
&
$mn$
\\
\midrule
\Resfield
\linebreak
of the theory
&
$\eta\spmatrix{0\\1}=u^{-1}\varphi$
&
$g_{\mu\nu}\dd x^\mu \otimes \dd x ^\nu=\eta(\beta,\beta)$
&
$\dottau$
\\
\textsl{dimension $(3)$}
&
$1$
&
$\frac{m(m+1)}2$
&
$0$
\\
\midrule
Degrees of freedom
\linebreak
of the theory
\linebreak
$(2)+(3)-(1)$
&
$4m$
&
$\frac{m(m^2+1)}2$
&
$mn$
\\
\bottomrule
\end{tabular}
\caption{Fields involved in the three examples described in the text, with their meanings and degrees of freedom, after applying the Main Lemma.}
\label{table-after-MainLemma}
\end{table}

In this section, we comment on the structures involved in our scheme applyied to the three main examples described above. We will designate by \caseone the example of the electroweak part of the Standard Model, \caseonebis its variant proposed in \ref{sec-an-alternative-dressing-field-particles}, \casetwo the example of the general relativity as a gauge theory, \casetwobis its variant proposed in \ref{sec-an-alternative-dressing-field-gravity}, and \casethree the gauge theory defined on Atiyah Lie algebroids. The reader is advised to look at Table~\ref{table-before-MainLemma} and Table~\ref{table-after-MainLemma} which fix the terminology we use in the following.

\subsection{About the dressing field}
\label{sec-about-the-dressing-field}

From the examples described in this paper, our scheme can be summarized as follows. A gauge theory contains a finite set of fields $\{\varphi_0, \varphi_1, \dots, \varphi_N \}$ in $\caG$-spaces. One of these fields, say $\varphi_0$, is chosen as the ``\auxfield'', and we decompose it into a couple $(\rho, u)$, where $u$ is the \dressfield of the Main Lemma, and $\rho$ is a residual field. The map $\varphi_0 \mapsto (\rho, u)$ has to be one-to-one and mathematically well-defined. In our examples, $u$ carries all the action of the gauge group that is factored out in the theory, so that $\rho$ is gauge invariant. The next step is to apply the Main Lemma to all the remaining fields in order to get gauge invariant composite fields $\hvarphi_i$, for $i=1, \dots, N$. 
The \auxfields identified in the three main examples are given in Table~\ref{table-before-MainLemma}. In \caseone, the scalar field $\varphi$ is added for phenomenological purposes in the model, while in \casetwo and \casethree, $\Lambda$ and $\tau$ are natural (mathematical) objects, which appear as components of extended notions of connections on Cartan geometries and on transitive Lie algebroids. 

The construction of gauge invariant fields proposed in \cite{Dira55a} and \cite{LaveMcMu97a} fit in our scheme. In these examples, the \auxfield is the connection $\omega$ itself, and the \dressfield $u$ is extracted from $\omega$ using a gauge-like constraint $\chi(\omega) = 0$ (similar to a gauge fixing) which consists to select a particular element in each gauge orbit (up to the Gribov ambiguity problem). This is not the usual gauge fixing procedure used to quantize the theory, in the sense that it does not consist to add a gauge fixing term in the Lagrangian: the procedure is applied at the level of $\caG$-spaces of fields. The \auxfield is then decomposed as $\omega \mapsto (\omega_\text{res}, u)$, where $\omega_\text{res}$, the \resfield, is the unique element of the gauge orbit of $\omega$ satisfying the constraint $\chi(\omega_\text{res}) = 0$, and where $u$ is computed, as a \emph{non local} expression in terms of $\omega$, such that $\omega = u \omega_\text{res} u^{-1} + u \dd u^{-1}$. Then, all the other fields of the theory, \textsl{i.e.} the matter fields, are mapped to gauge invariant composite fields. Notice that in \caseone, the decomposition $\varphi \mapsto (\eta, u)$ corresponds to the choice of the unique element $\eta \spmatrix{0 \\ 1}$ in the gauge orbit of $\varphi$, which solves the gauge-like constraint $\chi(\varphi) = \big\lVert \lVert \varphi \rVert \spmatrix{0 \\ 1} - \varphi \big\rVert = 0$ where $\lVert \cdot \rVert$ is the $\gC^2$-norm. The choice of this particular element in the gauge orbit depends on a reference configuration, here $\dotvarphi = \spmatrix{0\\ 1}$, see \ref{sec-dependence-on-a-reference-configuration}. Contrary to \cite{Dira55a} and \cite{LaveMcMu97a}, in our examples, the so-chosen \auxfield $\varphi_0$ is not the (Ehresmann) connection itself: this has the advantage to get \emph{local} expressions of $u$ in terms of $\varphi_0$.

In our scheme, the \dressfield $u$ yields a transfer of some degrees of freedom from the \auxfield to the connection $\omega$ giving rise to the composite field $\homega$, see the counting in Table~\ref{table-before-MainLemma} and in Table~\ref{table-after-MainLemma}. This transfer is reminiscent to the usual absorption of Goldstone scalar bosons in the spontaneous symmetry breaking mechanism. Indeed, in our scheme, the degrees of freedom which are absorbed in the gauge invariant composite fields correspond to the degrees of freedom neutralized in the symmetry group. However, \caseone shows that the strict identification between the  \dressfields and the Goldstone scalar bosons is not possible for several reasons. Firstly, the \dressfield is present from the beginning, before the symmetry neutralization, and independently of the sign of $\mu^2$. Secondly, the composite fields are gauge invariant, not massive fields: mass-like terms are generated, but in which $\eta^2$ replaces a constant mass parameter. An extra step is required to generate true mass terms when $\mu^2<0$, see \cite{MassWall10a} for details. A similar transfer of degrees of freedom can be checked in \casetwo. There, the \dressfield $\Lambda$ contains $m^2$ degrees of freedom: $\frac{m(m-1)}{2}$ degrees of freedom of $\Lambda$ correspond to the dimension of the neutralized group $SO(1,m-1)$, and the other $\frac{m(m+1)}{2}$ degrees of freedom enter in the definition of $g$ through the relation $g_{\mu\nu} = \eta_{ab} \Lambda^a{}_\mu \Lambda^b{}_\nu$. In \casethree, all the degrees of freedom of the group of symmetry are carried by the \dressfield $u$ in order to define the composite fields $\hA_\mu$.

\subsection{Dependence on a reference configuration}
\label{sec-dependence-on-a-reference-configuration}

In each case, the \dressfield is defined once a reference configuration in the space of the \auxfields is chosen. One wonders how the theory depends on this fixed reference configuration.

In \caseone, the question has been investigated in \cite{MassWall10a}: the unit vector $\dotvarphi = \spmatrix{0\\ 1}$ can be rotated by a constant element $v \in SU(2)$, and the new \dressfield $u'$ associated to $v \dotvarphi$ is related to the previous one by $u'=u v^{-1}$. Applying the Main Lemma, the new composite fields $B'_\mu$ are given by $v B_\mu v^{-1}$, which corresponds to a new global definition of the fields, in the same theory. Thanks to its $SU(2)$ invariance, the Lagrangian does not depend on $v$.  

In \casetwo, a change of the reference configuration $(\dd x^\mu)$ corresponds to a change of the coordinate system, for which $\dd x'^\mu=G_\nu^\mu \dd x^\nu$, with $G = \left( \frac{\dd x'^\mu}{\dd x^\nu}\right)$, so that $\Lambda'^a{}_\mu = \Lambda^a{}_\nu (G^{-1})_\mu^\nu$, and the new composite field $\hGamma'$ is related to $\hGamma$ by a Christoffel-like transformation. The Lagrangian being invariant by any change of coordinate systems, the theory is the same.

In \casethree, $\dottau$ can be mapped to $\dottau'=\Ad_v\dottau$, with $v\in\caG$, so that $u' = uv^{-1}$, and the new composite field is $\homega' = v\homega v^{-1} + v \hd v^{-1}$. This transformation is not a direct application of the Main Lemma since $\homega$ is not a connection $1$-form on the Lie algebroid $\lieA$, and $v$ cannot be identified as an \auxfield.  In the transformed Lagrangian, the field $v$ disappears, and the kinetic part of the fields $\homega'$ does not contain (constant) mass terms anymore: the mass-like terms depend on the non-constant reference configuration $\dottau'$. Obviously, the theories in terms of $\homega$ and $\homega'$ are equivalent. But, for at least two reasons, the reference configuration $\dottau = \Id_\lieL$ is a better and preferred physical parametrization. Firstly, it makes apparent the massive vector fields of the theory. Secondly, because, as noticed in section~\ref{sec-application-to-a-yang-mill-higgs-model-based-on-atiyah-lie-algebroids}, $\dottau = \Id_\lieL$ has the same local mathematical expression in any trivialization of $\lieL$, the masses obtained in the Lagrangian are the same in any trivialization, so that they are \emph{globally} defined.

It is tempting to consider the transformations from a reference configuration to another as gauge transformations, but this is not our point of view. An active gauge transformation should act on all the fields of the theory, which is not the case here. In the three examples, a gauge transformation of the \auxfield, in the original theory, is completely supported by the \dressfield $u$, so that the reference configuration is invariant. Moreover, a gauge transformation of the \dressfield is always of the form $u \mapsto \gamma^{-1} u$, while in the three cases, we have obtained some transformations of the form  $u \mapsto u v^{-1}$.  In \casetwo, the field $v = G$, being a change of coordinate system, is not an element of the gauge group. Finally, these transformations are not passive gauge transformations, in the sense defined in section~\ref{sec-introduction}, because it is not a change of local trivialization of the corresponding principal fiber bundle.

\subsection{Observables}
\label{sec-observables}

The Lagrangian of a gauge field theory supports two kind of symmetries: the gauge symmetry, and the symmetry under changes of coordinate system when the Lagrangian is written locally. Applying the Main Lemma, the gauge symmetry is neutralized, or only a part of it as in \caseone, in the sense that its action becomes trivial on the new fields of the theory. What remains is a theory which supports only the symmetry under changes of coordinate system. From this point of view, we have reduced a gauge field theory to a purely geometrical theory (or ``almost'' in \caseone), in the sense explained in section~\ref{sec-introduction}. This is crystal clear in \casetwo, where the final fields are those of the natural geometry of the base manifold $\varM$: the linear connection $\hGamma$ and the metric $g$. This reduction to ``more'' geometrical objects is meaningful in relation to physical observables. For instance, in \caseone, the composite fields $W^\pm_\mu$, $Z_\mu$ and $A_\mu$ are exactly the fields of the bosons experimentally detected, and the composite fermion fields give rise to the ``ordinary'' electron (\textsl{via} the combination of the left and right handed composite fields into a Dirac spinor field, see \cite{MassWall10a}). In \casetwo, the fields $\hGamma^\nu{}_{\rho\mu}$ are observable in a given coordinate system, when one uses the geodesic equation $\ddot{x}^\nu + \hGamma^\nu{}_{\rho\mu} \dot{x}^\rho \dot{x}^\mu = 0$ to measure them, exactly as we ``measure'' the components of a force vector in mechanics using the trajectory of a body under its influence. This is not the case for the original fields $\Gamma^a{}_{b \mu}$ and $\Lambda^a{}_\mu$. In the same way, in \casethree, we expect the fields $\hA_\mu$ to have the status of observables.

\smallskip
As can be noticed in Table~\ref{table-after-MainLemma}, the theories after the application of the Main Lemma exhibit some \resfields which deserve comments. A \resfield appears either as a byproduct of the decomposition of the \auxfield (\caseone) or when the \dressfield takes its values in a group $G$ larger than the structure group $H$ of the principal bundle (\casetwo and \caseonebis). Clearly this shows that the \resfield should carry the degrees of freedom of the \auxfield not involved in the corresponding \dressfield, as well as those of the \dressfield that are not transferred to the composite field. 
 
Nevertheless  \resfields are more difficult to identify than \dressfields. It does not seem possible to propose a systematic way to find them right away. However, given the criterion above, one could try to extract a \dressfield from the \auxfield as ``small'' as possible, that is, with values in the structure group $H$. The remaining degrees of freedom are then the wanted \resfield. \Caseone illustrates this possibility, but as shown in \casetwobis, this might not provide us with a globally well defined theory. 

It may be that the better strategy is to let the \resfield emerge as observable field directly from the Lagrangian after the change of variables\footnote{``Only the theory decides what can be observed'' once said Einstein to Heisenberg, as reported by the latter in \textsl{Physics \& beyond}.}.

\subsection{Relations to other symmetry reduction procedures}
\label{sec-relations-to-other-symmetry-reduction-procedures}

Our scheme is not a gauge fixing because it does not consist to add a constraint equation in the functional measure or in the Lagrangian of the theory. Rather, our scheme relies on a change of field variables at the level of the functional $\caG$-spaces of the theory. However, concerning the question of quantization, our procedure gives rise to a convenient result which is also the aim of some gauge fixing procedures: in \caseone and \casethree, the volume of the gauge group (or a part of it) can be factorized out in the functional integral. Indeed, after the change of field variables, the integrand depends only on gauge invariant fields, and not on the \dressfield, which captures the degrees of freedom of the gauge group. In this respect, our scheme answers, in an economic way, the question of the redundant gauge degrees of freedom, but its applicability is not universal: in \caseone, it remains to fix the $U(1)$-gauge symmetry. 

Relation between our scheme and the spontaneous symmetry breaking mechanism has already been discussed  in section~\ref{sec-the-electroweak-part-of-the-standard-model-of-particles-physics} and at the end of section~\ref{sec-about-the-dressing-field}.

Reductions of symmetries in the context of \caseone and \casetwo, have been formalized by some authors using reductions of principal fiber bundles. In \cite{Ster94a}, a map corresponding to our $\varphi$, which is called there a ``Higgs field'', is introduced, from which a map corresponding to our \resfield $\eta$, and a map corresponding to our \dressfield $u$, are extracted. This latter map performs the reduction $U(1)\times SU(2) \to U(1)$. In \cite{Trau79a, Neem79a, IvanSard81a} and more recently in \cite{Sard11a}, the metric $g$ is used to perform the reduction $GL^+_m(\gR) \to SO(1,m-1)$, and it is called a ``Higgs field'' in these papers. In both situations, the terminology ``Higgs field''  clearly designates different objects. Moreover, $\eta$ and $g$ have distinct mathematical status. But in our scheme, these two fields are \resfields of the neutralization procedure, so that they are ingredients of the same kind. It is satisfying that this formal analogy is also compatible with the physical fact that $g$ and $\eta$ are observables.

\section{Conclusion}

In this paper, we have put forward a scheme to construct gauge invariant composite fields from connections by performing a change of variables in the functional $\caG$-spaces of fields. As a result, the action of the gauge group on the new (composite and residual) fields is neutralized (trivial action of $\caG$), and this induces \textsl{de facto} a reduction of the symmetries of the theory. We have shown to what extent this scheme is different from other well-known methods of symmetry reductions: gauge fixing, spontaneous symmetry breaking, and reduction of principal fiber bundles. In order to illustrate this scheme, three main examples have been studied in details, providing a better understanding of this change of variables, in particular in terms of geometrization of gauge structures. It also makes apparent the specific role of the various fields involved in these theories, as well as their relationships.

In the first example, the composite fields give rise to the $Z_\mu$ and $W^\pm_\mu$ bosons of the Standard Model of particle physics. Relations with the usual spontaneous symmetry breaking mechanism have been discussed. In the second example, our procedure, applied to a $SO(1,m-1)$-gauge formulation of GR in terms of Cartan connections, allows us to construct the geometrical Christoffel symbols.  Comparisons have been made with other approaches which relate the gauge formulation and the geometric formulation of GR. The third example is a Yang-Mills-Higgs gauge theory, written in terms of generalized connections on Atiyah Lie algebroids, where the composite fields are massive vector bosons.

The diversity of the examples shows the versatility and the robustness of our scheme. Thus, we  expect that other relevant examples might be encompassed within our method. For instance, concerning the construction of the Wess-Zumino functional \cite{WessZumi71a}, formula~(4.33) in \cite{Zumi84a} and p.~164 in \cite{ManeStorZumi85a} suggest that our procedure could be applied in this context, so that the BRS treatment of anomalies could find a renewal. Another example may be found in \cite{Lorc13a} 
which is devoted to the problem of the proton spin decomposition. There, the group valued local field $U^{-1}_{\text{pure}}$ could be a candidate \dressfield extracted from the \auxfield $A_\text{pure} = U_{\text{pure}} \dd U^{-1}_{\text{pure}}$, see eqs.~(18) and (15), provided $A_\text{pure}$ can be  mathematically defined in a unique way in terms of the gauge field $A$.

\section*{Acknowledgments}

We would like to thank R.~Stora and Th.~Schücker for fruitful discussions.


\bibliography{bib-gauge-invariants}

\begin{thebibliography}{10}

\bibitem{EnglBrou64a}
F.~Englert and R.~Brout.
\newblock Broken symmetry and the mass of gauge vector mesons.
\newblock {\em Phys. Rev. Lett.}, 13:321--323, August 1964.

\bibitem{Higg64a}
P.~W. Higgs.
\newblock Broken symmetries and the masses of gauge bosons.
\newblock {\em Phys. Rev. Lett.}, 13:508--509, October 1964.

\bibitem{GuraHageKibb64a}
G.~S. Guralnik, C.~R. Hagen, and T.~W. Kibble.
\newblock Global conservation laws and massless particles.
\newblock {\em Phys. Rev. Lett.}, 13:585--587, November 1964.

\bibitem{KobaNomi96a}
S.~Kobayashi and K.~Nomizu.
\newblock {\em Foundations of Differential Geometry, vol. 1}.
\newblock Wiley Classics Library. Interscience Publishers, 1996.

\bibitem{KolaMichSlov93a}
I.~Kolar, P.~W. Michor, and J.~Slovak.
\newblock {\em Natural Operations in Differential Geometry}.
\newblock Springer-Verlag, 1993.

\bibitem{Trau79a}
A.~Trautman.
\newblock Fiber bundles, gauge fields, and gravitation.
\newblock In A.~Held, editor, {\em General Relativity and Gravitation},
  volume~1, page 287, New York, 1979. Plenum Press.

\bibitem{FourLazzMass13a}
C.~Fournel, S.~Lazzarini, and T.~Masson.
\newblock Formulation of gauge theories on transitive lie algebroids.
\newblock {\em J. Geom. Phys.}, 64:174--191, 2013.

\bibitem{RuegRuiz04a}
H.~Ruegg and M.~Ruiz-Altaba.
\newblock The {S}tueckelberg field.
\newblock {\em International Journal of Modern Physics A}, 19(20):3265--3347,
  2004.

\bibitem{Dira55a}
P.~A.~M. Dirac.
\newblock Gauge-invariant formulation of quantum electrodynamics.
\newblock {\em Canadian Journal of Physics}, 33(11):650--660, 1955.

\bibitem{Dira58a}
P.~A.~M. Dirac.
\newblock {\em The Principles of Quantum Mechanics}.
\newblock Oxford University Press, 4th edition, 1958.

\bibitem{LaveMcMu97a}
M.~Lavelle and D.~McMullan.
\newblock Constituent quarks from {QCD}.
\newblock {\em Physics Reports}, 279(1):1--65, 1997.

\bibitem{Ster94a}
S.~Sternberg.
\newblock {\em Group theory and physics}.
\newblock Cambridge University Press, 1994.

\bibitem{MassWall10a}
T.~Masson and J.-C. Wallet.
\newblock A remark on the spontaneous symmetry breaking mechanism in the
  standard model.
\newblock arxiv 1001.1176, 2010.

\bibitem{Fadd09a}
L.~D. Faddeev.
\newblock An alternative interpretation of the {W}einberg-{S}alam model.
\newblock In V.~Begun, L.~L. Jenkovszky, and A.~Polanski, editors, {\em
  Progress in High Energy Physics and Nuclear Safety}, NATO Science for Peace
  and Security Series -- B: Physics and Biophysics, pages 3--8. Springer, 2009.

\bibitem{CherFaddNiem08a}
M.~N. Chernodub, L.~D. Faddeev, and A.~J. Niemi.
\newblock Non-abelian supercurrents and de {S}itter ground state in electroweak
  theory.
\newblock {\em J. High Energy Phys.}, 2008(12):014, 2008.

\bibitem{Shar97a}
R.W. Sharpe.
\newblock {\em Differential Geometry, {C}artan's Generalization of {K}lein's
  {E}rlangen Program}, volume 166 of {\em Graduate Texts in Mathematics}.
\newblock Springer-Verlag, 1997.

\bibitem{GockSchu89a}
M.~G{\"o}ckeler and T.~Sch{\"u}cker.
\newblock {\em Differential Geometry, Gauge Theories, and Gravity}.
\newblock Cambridge University Press, 1989.

\bibitem{SimoChatSrin99a}
R.~Simon, S.~Chaturvedi, and V.~Srinivasan.
\newblock Congruences and canonical forms for a positive matrix: application to
  the {S}chweinler-{W}igner extremum principle.
\newblock {\em J. Math. Phys.}, 40(7):3632--3642, 1999.

\bibitem{BlytRobe02a}
T.~S. Blyth and E.~F. Robertson.
\newblock {\em Further linear algebra}.
\newblock Springer Undergraduate Mathematics Series. Springer-Verlag, London,
  2002.

\bibitem{SchwWign70a}
H.~C. Schweinler and E.~P. Wigner.
\newblock Orthogonalization methods.
\newblock {\em J. Math. Phys.}, 11:1693--1694, 1970.

\bibitem{ChatKapoSrin98a}
S.~Chaturvedi, A.~K. Kapoor, and V.~Srinivasan.
\newblock A new orthogonalization procedure with an extremal property.
\newblock {\em J. Phys. A}, 31(19):L367--L370, 1998.

\bibitem{LazzMass12a}
S.~Lazzarini and T.~Masson.
\newblock Connections on {L}ie algebroids and on derivation-based
  non-commutative geometry.
\newblock {\em J. Geom. Phys.}, 62:387--402, 2012.

\bibitem{Mack05a}
K.~Mackenzie.
\newblock {\em General Theory of {L}ie Groupoids and {L}ie Algebroids}.
\newblock Number 213 in London Mathematical Society Lecture Note Series.
  Cambridge University Press, 2005.

\bibitem{ORai86a}
L.~O'Raifeartaigh.
\newblock {\em Group Structure of Gauge Theories}.
\newblock Cambridge University Press, 1986.

\bibitem{Neem79a}
Y.~Ne'eman.
\newblock Gravity, groups, and gauges.
\newblock In A.~Held, editor, {\em General Relativity and Gravitation.},
  volume~1, page 309, New York, 1979. Plenum Press.

\bibitem{IvanSard81a}
D.~Ivanenko and G.~Sardanashvily.
\newblock Relativity and equivalence principles in a gauge theory of
  gravitation.
\newblock {\em Russ. Phys. J.}, 24:555--557, 1981.

\bibitem{Sard11a}
G.~Sardanashvily.
\newblock Classical gauge gravitation theory.
\newblock {\em Int. J. Geom. Methods Mod. Phys.}, 8(8):1869--1895, 2011.

\bibitem{WessZumi71a}
J.~Wess and B.~Zumino.
\newblock Consequences of anomalous {W}ard identities.
\newblock {\em Phys. Lett.}, B27:95, (1971).

\bibitem{Zumi84a}
B.~Zumino.
\newblock Chiral anomalies and differential geometry.
\newblock In B.~S. DeWitt and R.~Stora, editors, {\em Relativity, groups and
  topology~II}, Les Houches, Session XL, pages 1291--1322. Elsevier Science
  Publishers, 1984.

\bibitem{ManeStorZumi85a}
J.~Ma{\~n}es, R.~Stora, and B.~Zumino.
\newblock Algebraic study of chiral anomalies.
\newblock {\em Comm. Math. Phys.}, 102(1):157--174, 1985.

\bibitem{Lorc13a}
C.~Lorc{\'e}.
\newblock Geometrical approach to the proton spin decomposition.
\newblock {\em Phys. Rev. D}, 87:034031, 2013.

\end{thebibliography}

\end{document}